\documentclass[11pt]{article}
\usepackage[left=1in,right=1in,top=1in,bottom=1in]{geometry}

\usepackage{amsmath}
\usepackage{amsthm}
\usepackage{amssymb}

\usepackage{thmtools,thm-restate}

\usepackage{mdframed}

\newtheorem{definition}{Definition}

\newtheorem{theorem}[definition]{Theorem}

\newtheorem{lemma}[definition]{Lemma}

\usepackage{calc}

\newcommand{\ctilde}{\tilde{c}}
\newcommand{\ftilde}{\tilde{f}}
\newcommand{\gtilde}{\tilde{g}}
\newcommand{\htilde}{\tilde{h}}
\newcommand{\ntilde}{\tilde{n}}

\newcommand{\bbN}{\mathbb{N}}

\newcommand{\poly}{\mathrm{poly}}
\newcommand{\del}{\partial}

\newcommand{\cF}{\mathcal{F}}
\newcommand{\cP}{\mathcal{P}}

\newcommand{\cW}{\mathcal{W}}

\newcommand{\cY}{\mathcal{Y}}

\newcommand{\ThetaT}{\widetilde{\Theta}}
\newcommand{\BreakPU}{\mathrm{BreakPU}}
\newcommand{\BreakOW}{\mathrm{BreakOW}}
\newcommand{\Breaker}{\mathrm{Breaker}}
\newcommand{\Query}{\mathrm{Query}}
\newcommand{\QueryX}{\mathrm{QueryX}}
\newcommand{\QueryY}{\mathrm{QueryY}}

\newcommand{\SafeToAnswer}{\mathrm{SafeToAnswer}}
\newcommand{\BuildPolynomial}{\mathrm{BuildPolynomial}}

\newcommand{\Pd}{\mathsf{P}}
\DeclareMathOperator*{\E}{E}

\newcommand{\centerComment}[1]{\hfill/\!\!/\begin{minipage}[t]{0.5\linewidth}{#1}\end{minipage}}

\title{Constructing a Pseudorandom Generator
Requires an Almost Linear Number of Calls}

\author{Thomas Holenstein\thanks{ETH Zurich, \texttt{thomas.holenstein@inf.ethz.ch}. This work was supported
by the Swiss National Science Foundation Grant No.~200021-132508} \and Makrand Sinha\thanks{University of Washington, \texttt{makrand@cs.washington.edu}. Parts of this work was done while the author was a student at ETH Zurich. Work supported by the Excellence Scholarship and Opportunity Programme of the ETH Zurich Foundation.}}
\begin{document}

\maketitle

\begin{abstract}
We show that a black-box construction of a pseudorandom generator 
from a one-way function needs to make $\Omega(\frac{n}{\log(n)})$
calls to the underlying one-way function.
The bound even holds if the one-way function
is guaranteed to be regular.  
In this case it
 matches the best known construction due to Goldreich,
Krawczyk, and Luby (SIAM J.~Comp.~22, 1993), which uses 
$O(\frac{n}{\log(n)})$ calls.
\end{abstract}

\section{Introduction}
\subsection{One-way functions and pseudorandom generators}
Starting with the seminal works by Yao \cite{Yao82}, 
and Blum and Micali \cite{BluMic84}, researchers have 
studied the relationship between various cryptographic
primitives, such as one-way functions, pseudorandom generators,
pseudorandom functions, and so on, producing a wide variety of results.
One particular task which was achieved was the construction of pseudorandom
generators from one-way functions, a task which has a history on
its own.
First, it was shown that one-way permutations imply
pseudorandom generators \cite{Levin87,GolLev89}.
Later, the result was extended to regular one-way functions
\cite{GoKrLu93}, and finally it was shown that arbitrary
one-way functions imply pseudorandom generators \cite{HILL99}.

Unfortunately, the constructions given in \cite{GoKrLu93} 
and \cite{HILL99} are relatively inefficient (even though 
they run in polynomial time).
Suppose we instantiate the construction
given in \cite{GoKrLu93} with a regular one-way
functions taking $n$ bits to $n$ bits.
Then, it yields a 
pseudorandom generator whose input is of 
length\footnote{The
$\ThetaT$-notation ignores poly-logarithmic factors.} 
$\ThetaT(n^3)$
and calls the underlying one-way function $\ThetaT(n)$ times.
The parameters in \cite{HILL99} are worse: if we instantiate
the construction with an (arbitrary) one-way function
taking $n$ bits to $n$ bits, we obtain a pseudorandom generator
which needs $\ThetaT(n^8)$ bits of input, and 
which does 
around\footnote{We counted $\ThetaT(n^{12})$ calls,
but since \cite{HILL99} is not completely explicit
about the construction, we make no guarantee that other
interpretations are impossible\ldots}
$\ThetaT(n^{12})$ calls to $f$.
The parameters of the security reduction are also very weak.

Naturally, many papers improve the efficiency of these results:
\cite{HaHaRe06a, Holens06b} show that the result of
\cite{HILL99} can be achieved with a more efficient reduction
in case one assumes that the underlying one-way function has 
stronger security than the usual polynomial time security.
\cite{HaHaRe06b} reduces the input length of the pseudorandom
generator in \cite{GoKrLu93} to $\Theta(n\log(n))$.
Also it reduces the input
length in \cite{HILL99} by a factor of $\ThetaT(n)$, and the number
of calls by a factor of $\ThetaT(n^3)$.
Most impressive, \cite{HaReVa10} reduces the 
seed length to $\ThetaT(n^4)$ and the number of calls to 
$\ThetaT(n^3)$, for the construction of a pseudorandom generator
from an arbitrary one-way function.
Finally, \cite{VadZen12} reduce the seed length in this
last construction to $\ThetaT(n^3)$.

We remark that the main focus on the efficiency has been on reducing
the seed length.
This is reasonable, as (private) randomness is probably the most
expensive resource.\footnote{We would like to mention that in part
this focus also seems to come
from the (somewhat arbitrary) fact
that people usually set the security parameter equal to the input length.
For example, suppose we have a one-way function from $n$ to $n$ bits
with security $2^{n/100}$ (meaning that in time $2^{n/100}$ one can
invert $f$ only with probability $2^{-n/100}$).
If a construction now yields a pseudorandom generator with $m = n^2$
bits of input, the security can at most be $2^{\sqrt{m}/100}$.
At this point it becomes tempting to argue that because
$m \mapsto 2^{\sqrt{m}/100}$ is 
a much slower growing function than $n \mapsto 2^{n/100}$,
it is crucial to make the input length as small as possible.
However, if one introduces a security parameter $k$, 
both primitives could have security roughly $2^{k}$.
Arguing over the function which maps the input length to the
security is not a priori a good idea.} 
Nevertheless, one would like both the seed length and the number
of calls to be as small as possible.
\subsection{Black-box separations}
After \cite{BluMic84,Yao82}, it was natural to try and 
prove that one-way functions do imply seemingly
stronger primitives, such as key agreement.
However, all attempts in proving this failed, and so 
researchers probably wondered (for a short moment) whether in fact
one-way functions \emph{do not} imply 
key agreement.
A moment of thought reveals that this is unlikely to be true: 
key-agreement schemes seem to exist, and so in fact we believe
that---consider the following as a purely logical 
statement---one-way functions \emph{do} imply key-agreement.

A way out of the dilemma was found by Impagliazzo and 
Rudich in a break through work \cite{ImpRud89}. 
They observed that the proofs of 
most results such as ``one-way functions
imply pseudorandom generators'' are, in fact, much stronger.
In particular, the main technical part of \cite{HILL99} 
shows that there exists 
oracle algorithms $g^{(f)}$ and $A^{(\Breaker,f)}$ with the following
two properties:
\begin{itemize}
\item For any oracle, $g^{(f)}$ is an expanding function.
\item For any two oracles $(\Breaker,f)$, 
if $\Breaker$ distinguishes the output
of $g^{(f)}$ from a random string, then $A^{(\Breaker,f)}$ inverts $f$.
\end{itemize}
Impagliazzo and Rudich then showed that the analogous statement
for the implication ``one-way functions imply key-agreement''
is simply wrong, giving the first ``black-box separation''.

After the paper of Impagliazzo and Rudich, many
more black box separations have been given (too many to list them all).
We use techniques from several papers:
in order to prove that there is no black-box construction
of collision resistant hash-functions from one-way permutations,
Simon \cite{Simon98} introduced the method of giving specific oracles
which break the primitive to be constructed.
Such oracles (usually called $\Breaker$) are now widely used, including
in this paper.
Gennaro et al.~\cite{GGKT05} developed an ``encoding paradigma'', a
technique which allows to give very strong black-box separations, 
even excluding non-uniform security reductions.
This encoding paradigma has first 
been combined with a Breaker oracle in \cite{HHRS07}.
In \cite{HaiHol09} a slightly different extension of \cite{Simon98} 
is used: their technique analyzes how $\Breaker$ behaves in case
one modifies the given one-way permutation on a single 
randomly chosen input.
We also use this method.

Some black box separation results are (as we are) concerned
with the efficiency of constructing pseudorandom generators.
Among other things, 
Gennaro et al.~\cite{GGKT05} show that in order to get a pseudorandom
generator which expands the input by $t$ bits, a black-box
construction needs to do at least $\Omega(t/\log(n))$ 
calls to the underlying one-way function (this matches 
the combination of Goldreich-Levin \cite{GolLev89} with
the extension given in 
Goldreich-Goldwasser-Micali \cite{GoGoMi86}).
In \cite{Viola05}, Viola shows that in order for
a black-box construction to expand the input by $t$
bits, it needs to do at least one of 
(a) adaptive queries, (b) sequential computation, or
(c) use $\Omega(t\cdot n\log(m))$ bits of input, when
the underlying one-way function maps $n$ to $m$ bits.
This result has been somewhat strengthened by Lu \cite{Lu06}.
The papers \cite{BrJuPa11,MilVio11} both study how much
the stretch of a given generator can be enlarged, as long
as the queries to the given generator are non-adaptive.

\subsection{Contributions of this paper}
A natural question to ask is: ``what is the minimum seed length
and the minimal number of calls needed
for  a black-box construction of a pseudorandom generator
from a one-way function?''

To the best of our knowledge, it is 
consistent with current knowledge that a 
construction has seed length $\Theta(n)$ and does a single call
to the underlying one-way function (however, recall that
\cite{GGKT05} show that in order to get a stretch of $t$
bits, at least $\Omega(t/\log(n))$ calls need to be made).

The reason why no stronger lower bounds are known
seems to be that from a one-way \emph{permutation}
it is possible to get a pseudorandom generator very efficiently
by the Goldreich-Levin theorem \cite{GolLev89}: 
the input length only doubles, and the construction 
calls the underlying one-way permutation once.
Also, almost 
all black-box separation results which prove
that a primitive is unachievable from one-way functions also 
apply to one-way permutations.
The only exceptions to this rule we are 
aware of is given by \cite{Rudich88,KaSaSm11}
where it is shown that one-way permutations cannot be 
obtained from one-way functions, and \cite{MatMat11}, where this
result is strengthened. 
However, both these results use a technique which does not
seem to apply if one wants to give lower bounds
on the efficiency of the construction of pseudorandom
generators.\footnote{Both proofs use
the fact that a one-way permutation satisfies
$g(v) \neq g(v')$ for any $v \neq v'$ crucially.}

One should note that a very efficient construction
of a pseudorandom generator from a one-way function
might have implications for practice: it is not inconceivable
that in this case, practical 
symmetric encryption could be based on a one-way function,
at least in some special cases where one would like a very
high guarantee on the security.

We show in this paper than any construction must make
at least $\Omega(\frac{n}{\log(n)})$ calls to the underlying one-way function.
While this bound is interesting even for arbitrary one-way functions,
it turns out that our proof works with some additional work 
even if the one-way function is guaranteed to be regular.
In this case, the number of calls matches the parameters in \cite{GoKrLu93}
(and recall that the length of the seed has been reduced to $O(n \log(n))$
in \cite{HaHaRe06b}, with the same number of calls to the one-way function).

In our theorem, we exclude a \emph{fully black-box reduction}, 
using the terminology of 
\cite{ReTrVa04}. 
In fact, we give three results.

In our first result, we assume that the construction $g(\cdot)$
when used with security parameter $k$ \emph{only calls the 
underlying one-way function with the same security parameter $k$}.
We believe that this is a natural assumption, as all constructions
we know have this property, and the underlying input length
is not immediately defined if $g$ makes calls to $f(k,\cdot)$ for
various values of $k$.
The result is stated in Theorem~\ref{thm:main}.

Next, we study black-box constructions with the same restriction, 
but where the security reduction is non-uniform. 
These can be handled with the technique from \cite{GGKT05},
and in our case it yields Theorem~\ref{thm:mainB}.

Finally, we remove the restriction that the 
construction calls the underlying function with a fixed security
parameter.
This gives Theorem~\ref{thm:mainC}.
However, one needs to be careful somewhat, since in this
case, the construction calls the given one-way function on 
a number of input lengths $n$, and thus already the 
expression $\Omega(n/\log(n))$ in our lower bound needs to be specified
more exactly.
Our theorem uses \emph{the shortest input length} of any call to $f$ (i.e.,
our lower bound is weakest possible in this case).
Also, we remark that this last bound does not exclude the construction
of ``infinitely often pseudorandom generators'', which are secure
only for infinitely many security parameters.

\section{The Main Theorem}

We think of a one-way function as a family $\{f_{k}\}_{k \geq 0}$,
indexed by some security parameter $k$.
The function $f_k$ then takes as input a bitstring of length $n(k)$,
and outputs a bitstring of length $n'(k)$.
Usually, the case $n(k) = n'(k) = k$ is considered in the literature.
We want to distinguish $n$ and $k$ here, as we hope this makes
the discussion clearer.
However, we will still require that $n$ is polynomially related 
to $k$.\footnote{The requirement that $n(k) \leq k^c$ is implicit
in the definition of one-way functions, as 
otherwise the one-way function cannot be evaluated in
time polynomial in $k$.
The requirement $n(k) \geq k^c$ is different, however.
For example, suppose a family $\{f_k\}_{k\geq0}$ can be evaluated in
time $k^{O(1)}$ and has $n(k) = \log^2(n)$.
Also, suppose that $f_k$ is a one-way function in the sense that
in time $k^{O(1)}$ it cannot be inverted with 
probability $k^{-O(1)} \leq 2^{-\sqrt{n(k)}}$.
If $f$ is additionally regular, fewer than $\Omega(n/\log(n))$ calls
are sufficient to construct a pseudorandom generator.}

\begin{definition}
A function $n(k) : \bbN \to \bbN$ is a \emph{length function}
if there exists $c \in \bbN$ such that
$k^{1/c} \leq n(k) \leq k^c$, $n(k)$ can be computed in 
time $k^c$, and $n(k+1) \geq n(k)$ for any~$k$.
\end{definition}

In general, the length $n(k)$ of the input of a one-way function differs
from the length $n'(k)$ of the output.
In case $n(k) > n'(k)$, it is shown in \cite{DeHaRe08} how to obtain a 
``public-coin collection of one-way functions'', where both the input
and the output length are $n'(k)$.
Such a collection can be used with known constructions to get a pseudorandom
generator, and the number of calls will only depend on $n'(k)$.
In case $n(k) < n'(k)$, it is easy to see that one can also get
a``public-coin collection of one-way functions'' with input and output
length $2n(k)$.

Therefore, we can restrict ourselves to the case $n(k) = n'(k)$, and 
see that otherwise, the parameter $\min(n(k),n'(k))$ is the quantity of
relevance to us.

\begin{definition}\label{def:owfprg}
A \emph{one-way function} $f = \{f_{k}\}_{k \geq 0}$ is a family
of functions $f_k: \{0,1\}^{n(k)} \to \{0,1\}^{n(k)}$, computable in
time $\poly(k)$, such that for any algorithm $A$ running in 
time $\poly(k)$ the function mapping $k$ to 
\begin{align}
\Pr_{x,A}[\text{$A(k,f_k(x))$ inverts $f_k$}]
\end{align}
is negligible 
in~$k$.\footnote{We say that ``$A(f_k(x))$ inverts $f_k$'' 
if $f_k(A(f_k(x))) = f_k(x)$, and write $A$ below 
the symbol $\Pr$ to indicate that the probability
is also over any randomness $A$ may use.
We also assume it is clear that $x$ is picked from $\{0,1\}^{n(k)}$.}

A \emph{pseudorandom generator} $g = \{g_{k}\}_{k\geq 0}$ is a 
family of polynomial time computable functions $g_k: \{0,1\}^{m(k)}
\to \{0,1\}^{m'(k)}$ with $m'(k) > m(k)$ and such that any algorithm $B$
running in time $\poly(k)$
\begin{align}
\Pr_{v,B}[B(k,g_k(v))=1] - \Pr_{w,B}[B(k,w)=1]
\end{align}
is negligible in $k$.
\end{definition}

We next define \emph{fully black-box constructions}, but
only for the special case of importance to us.
Note that we assume that the underlying one way function is regular 
(a function family $\{f_{k}\}_{k\geq0}$ is regular if
$|\{x' : f_{k}(x') = f_{k}(x)\}|$ only depends on $k$ and not on $x$).
\begin{definition}\label{def:fbb}
A \emph{fully black-box construction of a pseudorandom generator 
from a regular one-way function} consists of two oracle algorithms $(g,A)$.
The \emph{construction} $g^{(f)}$ is a polynomial time 
oracle algorithm which provides, for each length function $n(k)$
and each $\ell$, 
a function $g_\ell: \{0,1\}^{m(\ell)}
\to \{0,1\}^{m'(\ell)}$ with $m'(\ell) > m(\ell)$.
For this, $g$ may call $f$ as an oracle.

Further, the \emph{security reduction} $A^{(\cdot,\cdot)}(k,\cdot,\cdot)$ 
is a $\poly(k,\frac{1}{\epsilon})$-time oracle
algorithm such that for any regular function $f$, any inverse
polynomial function $\epsilon(\ell)$, and any 
oracle $\Breaker$ for which
\begin{align}\label{eq:1000}
\Pr_{v,\Breaker}[\Breaker(\ell,g_\ell(v))=1] - 
\Pr_{w,\Breaker}[\Breaker(\ell,w)=1] 
\geq \epsilon(\ell)
\end{align}
for infinitely many $\ell$, then
\begin{align}\label{eq:1001}
\Pr_{x,A}[\text{$A^{(\Breaker,f)}(k,\epsilon(k),f_k(x))$ inverts $f_k$}]
\end{align}
is non-negligible.
\end{definition}

In a large part of the paper we restrict ourselves to the (most interesting)
case where $g$ only calls $f$ on a single security parameter.

\begin{definition}
A black-box construction is \emph{security parameter restricted}
if $g(k,\cdot)$ only
calls $f(k,\cdot)$ and $A(k,\cdot)$ only calls $\Breaker(k,\cdot)$ and
$f(k,\cdot)$ for any $k$.
\end{definition}

Our main contribution is the following theorem:
\begin{restatable}{theorem}{maintheorem}\label{thm:main}
Let $n(k),r(k)\in \poly(k)$ be computable in time $\poly(k)$,
and assume that 
$r(k) \in o(\frac{n(k)}{\log(n(k))})$.
There exists no security parameter restricted
fully black-box construction of
a pseudorandom generator from a one-way function 
which has the property that $g(k,v)$ does at most $r(k)$ 
calls to $f(k,\cdot)$.
\end{restatable}

The above discussion assumes that the adversary is uniform
(i.e., there is a single adversary $A^{(\cdot,\cdot)}$ with oracle
access to $f$ and $\Breaker$).
However, many black-box results even work in case that
$A$ can be a non-uniform circuit, and our result is no exception.
We define non-uniform black-box constructions in Section~\ref{sec:nubb},
and then prove the following theorem (we also change the security 
of the one-way function from standard security to security $s(k)$
in order to illustrate what results we can get in this case).

\begin{restatable}{theorem}{maintheoremB}\label{thm:mainB}
Let $r(k)$, $s(k)$, $n(k)$ be given, 
and assume  $r(k) < \frac{n(k)}{1000 \log(s(k))}$
for infinitely many $k$.
Then, there is no \emph{non-uniform} security parameter
restricted fully black-box 
construction of a pseudorandom generator from a one-way function
with security $s$ which has the property that $g(k,v)$ does at most
$r(k)$ calls to $f(k,\cdot)$.
\end{restatable}

In Section \ref{sec:nsprc} we study what happens with 
black-box constructions which are not security paramter
restricted.
To explain our results in this setting, we need a few more definitions.
Suppose we have given an oracle construction $(g,A)$, and
fix the oracle $f$ (i.e., the one-way function).
For each $\ell$ we then consider the \emph{shortest} call which
$g(\ell,v)$ makes to~$f$ for any~$v$:
\begin{align}
n_f^{-}(\ell) := \min \{n(k) | 
\text{$\exists v: g^{(f)}(\ell, v)$ queries $f(k,\cdot)$}\}.
\end{align}
Analogously, for each $\ell$ we consider the maximal number of
calls $g(\ell,v)$ makes to $f$:
\begin{align}
r_f(\ell) := \max \{r | 
\text{$\exists v: g^{(f)}(\ell, v)$ makes $r$ queries to $f$}\}.
\end{align}
Note that both $n_f^{-}$ and $r_f$ do in general depend on the oracle $f$.

Our second main theorem is then given in the following: 

\begin{restatable}{theorem}{maintheoremC}\label{thm:mainC}
Fix a length function $n(k)$.
Let $(g,A)$ be a fully black-box construction
of a pseudorandom generator from a regular one-way function.
Then, there is an oracle $f$ for which
\begin{align}
r_f \in \Omega\Bigl(\frac{n^{-}_f}{\log(n^{-}_f)}\Bigr)\;.
\end{align}
\end{restatable}

\section{Notation and Conventions}

In most of the paper, we consider one fixed security parameter $k = \ell$.
Then, the input length $n = n(k)$ of the one-way function
and the input length $m = m(k)$ of the pseudorandom generator are also
fixed.

\subsection{Pseudouniform functions}

A pseudouniform function is a family $g = \{g_{k}\}_{k\geq 0}$
of length 
preserving functions $g_k : \{0,1\}^{m(k)} \to \{0,1\}^{m(k)}$
such that the output of $g_k$ is indistinguishable from a uniform
string.
An example is given by the identity function, or
any one-way permutation.

\begin{definition}
A function family $g = \{g_{k}\}_{k\geq 0}$ where
$g_k: \{0,1\}^{m(k)} \to \{0,1\}^{m(k)}$ of $\poly(k)$-time
 computable functions is pseudouniform if, for all algorithms
$A$ running in time $\poly(k)$ the function
\begin{align}
\Bigl|\Pr_{A,v}[A(k,g_k(v)) = 1] - \Pr_{A,w}[A(k,w)] = 1\Bigr|
\end{align}
is negligible in $k$.
\end{definition}

If we are given a family $\{g_{k}\}_{k\geq0}$ which is both
pseudouniform and a one-way function, then we can obtain
a pseudorandom generator using only one call to $g$
by the Goldreich-Levin Theorem \cite{GolLev89}.
Conversely, given a pseudorandom generator one can 
get a pseudouniform one-way function by truncating the output.

\begin{theorem}\label{thm:prgandPUOWF} 
Suppose that $g = \{g_k\}$ is both a pseudouniform function 
and also a one-way function.
Then, $h_{k}(v,z) := (g(v),z, \oplus_{i=1}^n v_i z_i)$ 
is a pseudorandom generator.

Conversely, if $g$ is a pseudorandom generator with $m(k)$
bits of input, the
truncation of $g$ to the first $m(k)$
bits of its output is both pseudouniform and a one-way function.
\end{theorem}
\begin{proof}
The first part follows immediately by the fact that a
distinguisher can be converted to a next bit predictor \cite{BluMic84}
and the Goldreich-Levin Theorem \cite{GolLev89}.

For the second part, let $g: \{0,1\}^{m(k)} \to \{0,1\}^{m(k)+1}$ 
be a pseudorandom generator
where we assume without loss of generality that $g$ expands
by $1$ bit.
If the truncation $g': \{0,1\}^{m(k)} \to \{0,1\}^{m(k)}$ is
not pseudouniform, there must be some distinguisher which has
non-negligible advantage in distinguishing the output
from a uniform random string.
Such a distringuisher immediately contradicts the pseudorandomness
of $g$.

Suppose now that $g': \{0,1\}^{m(k)} \to \{0,1\}^{m(k)}$
is not a one-way function.
Then, there exists some (inverse) polynomial $\epsilon(k)$
and some algorithm $A$ which inverts $g$ with probability 
at least $\epsilon$ for infinitely many $k$.

On some fixed security parameter $k$ we now proceed as follows:
first, let $p$ be the probability
that $A$ finds a preimage of $g'$ of a uniformly chosen 
element $y \in \{0,1\}^m$ (i.e., 
the probability that $g'(A(y)) = y$ for a uniform
random $y$).
This can be arbitrary small, because the distribution
is different from the distribution induced by $g'(x)$.
Using sampling, we can find an estimate 
$p'$ of $p$ such that with probability 
$1-2^{-k}$ the estimate satisfies $|p-p'| \leq \epsilon/4$.
If $p' \leq \epsilon/2$, then we can distinguish the output 
of $g$ from uniform by checking whether $A$ inverts $g'$
on the first $m(k)$ bits.  

On the other hand, if $p' \geq \epsilon/2$ we can assume 
$p \geq \epsilon/4$.  
Now $A$ immediately gives an inverter for $g$ which inverts a random
uniform bitstring of length $m+1$ with probability at 
least $\epsilon/8$ (just ignore the last bit, invert
$g'$, and hope the last bit matches).
Finally, the probability an inverter for $g$ inverts an
output of $g$ is at least twice the probability it inverts 
a uniform bitstring.
Thus, we can get a distinguisher by checking whether $A$ even finds
an inverse of $g$, given only the first $m$ bits of the result.
\end{proof}

Thus, we see that giving lower bounds on the construction of
pseudorandom generators is equivalent to giving lower bounds on 
the construction of pseudouniform one-way functions.

\subsection{Normalization}

Suppose we have a construction $\{g_k^{(f)}\}_{k \geq 0}$
of a supposedly pseudouniform one-way functions, where 
$k$ is a security parameter.  
We make several assumptions on the construction which
simplifies the proofs.
First, we assume that $g$ never calls $f$ twice with the same input,
and does exactly $r$ calls to $f$.
This is easy to achieve:
one can modify $g$ to get an equivalent oracle construction
with these properties.
Next, we enlarge the range of $g$, and assume that in case
two queries of $f$ give the same \emph{answer}, then $g$ outputs a special 
symbol which encodes a failure.
This last restriction is not completely trivial, as it
can break some constructions of pseudouniform functions
for some choices of underlying one-way functions.
As we will see in the proof of Theorem~\ref{thm:main}, in
our case this is no problem (because of the way we construct
the oracles $f_k$).
\begin{definition}
Let $\{0,1\}^{m*} := \{0,1\}^{m} \cup \{(\bot, v) | v \in \{0,1\}^{m}\}$.
An oracle function $g^{(f)}: \{0,1\}^{m} \to \{0,1\}^{m*}$ is
$r$-query normalized if $g(v)$ never queries $f$ with the same input twice,
does exactly $r$ calls to $f$, and
 whenever two outputs of $f$ agree, $g^{(f)}(v) = (\bot, v)$.
\end{definition}
We will write $g$ instead of $g^{(f)}$ whenever $f$ is clear
from the context.
Furthemore, we let $g'(v,y_1,\ldots,y_r)$ be the function 
which never calls $f$ but instead just uses $y_i$ as the 
reply of $f$ to the $i$th query.

\subsection{Notations}
\label{sec:notations}
\begin{definition}[The Query-sets]
The set $\Query(g,v,f)$ is
$\{(x_1,y_1),\ldots,(x_r,y_r)\}$,
where $x_i$ is the $i$-th query which $g$ does to $f$
in an evaluation of $g^{(f)}(v)$, and $y_i$ is the answer
given by $f$.
The set $\Query(g',v,y_1,\ldots,y_r))$ is defined similarly 
(in particular, it also contains pairs $(x_i,y_i)$).
The sets $\QueryX(g,v,f)$ and $\QueryY(g,v,f)$ contain the 
$x$ and $y$-part of the pairs in $\Query(g,v,f)$.
\end{definition}

For a pair $(x^*,y^*)$, we 
define
\begin{align}\label{eq:6}
f_{(x^*,y^*)}(x) := \begin{cases}
y^* & \text{if $x = x^*$}\\
f(x) & \text{otherwise.}
\end{cases}
\end{align}

We use the following sets of functions $f: \{0,1\}^n \to
\{0,1\}^n$.
For a set $\cY \subseteq \{0,1\}^n$ such that $|\cY|$ divides
$2^n$, $\cF(\cY)$ is the set of all regular surjective functions
$f: \{0,1\}^n \to \cY$.
Then, $\cP_n$ is the set of all bijective functions
$f: \{0,1\}^n \to \{0,1\}^n$, i.e., the permutations.
We use $\cP$ instead of $\cP_n$ when $n$ is clear
from the context, and write $f \leftarrow \cP_n$ or $f \leftarrow \cF(\cY)$ to pick
a function uniformly from the respective set.

\section{Overview of the Proofs}
We now try to provide some intuition of the proofs.
We concentrate on the proof of Theorem~\ref{thm:main}, and only
say a few words about the other theorems in the end.

\paragraph{Basic setting}
By the discussion above, it is sufficient to consider constructions
of pseudouniform one-way functions from one-way functions.
Thus, suppose a fully black-box construction $(g,A)$ 
of a pseudouniform one-way function is given.
We fix some security parameter $k$, and consider $g(k,\cdot)$,
which only calls $f(k,\cdot)$.

Our task is to come up with a pair $(\Breaker,f)$, such that
$\Breaker(k,\cdot)$ either inverts $g$ or distinguishes the output
of $g$ from a uniform random string, and yet $A^{(\Breaker,f)}$ will
not invert $f(k,\cdot)$ with noticeable probability.

\subsection{The case of a single call}
We first study the case where $g^{(f)}$ does a single call to the
underlying one-way function.

\paragraph{Example constructions}
We first discuss three example constructions for $g^{(f)}$, which all 
do $r=1$ calls to $f$.

The first example $g: \{0,1\}^n \to \{0,1\}^n$ is defined as
$g(v) = f(v)$, i.e., the function simply applies the given one-way 
function.
Clearly, $g$ will be one-way, so that $\Breaker$ must distinguish
the output of $g$ from a random function; we will call such a 
breaker $\BreakPU$.
In this case, our proof will pick $f: \{0,1\}^n \to \{0,1\}^n$ as a 
very degenerate function (for example with image set of 
size $|\cY| = 2^{\log^2(n)}$).
It is intuitive that $\BreakPU$ can distinguish the output of $g$ from 
a uniform random string without helping to invert $f$.

The second example $g: \{0,1\}^n \to \{0,1\}^n$ is defined as 
$g(v) = v$, so that the function simply outputs the input $v$.
In this case, clearly the function is pseudouniform, therefore
$\Breaker$ will break the one-way property of $g$ using exhaustive
search.  We will call such a breaker $\BreakOW$.

The last example $g: \{0,1\}^{2n} \to \{0,1\}^{2n}$ is defined
as \begin{align}
g(v,r) := \begin{cases}
(v,r) & \text{if $r \neq 0^n$}\\
(f(v),r) & \text{otherwise}.
\end{cases}
\end{align}
This function is pseudouniform no matter how $f$ is defined.
Thus, $\Breaker$ needs to invert $g$.
One sees that it needs to be careful in that: if $\BreakOW(y,0^n)$
returns a preimage of $g$, clearly $A$ will be able to invert.
Thus, only images $(y,r)$ with $r \neq 0^n$ should be inverted.

\paragraph{Inverting constructions with one call}
It turns out that we can describe $\BreakOW(w)$ in general
as follows: enumerate all possible inputs~$v$, and 
evaluate $g^{(f)}(v)$ on each of them.
In case $g^{(f)}(v) = w$, 
 $\BreakOW$ considers the output $y$ which appeared in this
evaluation as answer to the query done to $f$.
It then considers the probability that $w$ is the output 
in case nothing about $f$ or $v$ is known, 
but conditioned on $y$ to appear in the evaluation 
(assuming that $f$ is chosen as a permutation).
If this probability is large (concretely, larger than $2^{-m+n/30}$), 
$\BreakOW$ refuses to answer. Otherwise, it returns $v$.

A very quick intutition why this might not help to invert
$f$ is as follows: suppose
an algorithm $A^{(\BreakOW,f)}(y)$ tries to invert $y$.
In order to do use $\BreakOW$, $A$ needs to find some useful $w$
for this $y$.
However, $\BreakOW$ ensures that it only inverts $w$ which
are not very likely to be outputs for this $y$, so that
$A$ is unlikely to find a matching $w$.
Thus, we can hope that $A$ will fail.

We will sketch the actual proof that $f$ remains one-way given $\BreakOW$
later.

\paragraph{Invert or distinguish constructions with one call}
We now distinguish two cases: if $\BreakOW$ inverts $g$ for
a randomly chosen $f$ with probability
(say) $\frac12$, clearly we are done.
Otherwise, it must be that very often in random evaluations of
$g$, once the output of $f$ is fixed to $y$, certain
values $w$ are much more likely (if the rest, i.e., $v$, and $f$,
are still chosen at random).
In this case, we first pick $\cY \subseteq \{0,1\}^n$,
$|\cY| = 2^{\log^2(n)}$ as image set.
We then show that there is some small set $W(\cY)$ depending only
on $\cY$, such that if we pick $f$ from $\cF(\cY)$ and $v$ uniformly
at random, with high probability $g^{(f)}(v) \in W$.
Thus, we can distinguish the output of $g^{(f)}(v)$ from a uniform
random string by just checking whether it is in $W$, and this
without even knowing the details of $f$ (namely, we can still
pick $f: \{0,1\}^n \to \cY$ uniformly at random).

The reason that $g^{(f)}(v) \in W(\cY)$ is likely should be intuitive:
we know that conditioning on some fixed $y$ highly biases the output
$w$, and because there are only few $y \in \cY$, the output should
still be biased overall.

\paragraph{The underlying one-way function remains one-way}
We still need to argue that $\BreakOW$ does not help to invert
a random permutation $f$.
For this, suppose $A^{(\BreakOW,f)}(y_0)$ tries to invert $y_0 = f(x_0)$.
Pick a random $x^*$ and consider the function $f^* = f_{(x^*,y_0)}$, 
as defined in Section~\ref{sec:notations}.
Also, let $\BreakOW^*$ be defined as $\Breaker$, except that it uses
$f^*$ instead of $f$ when it evaluates $g$ in the exhaustive search.

Intuitively, if $A^{(\BreakOW,f)}(y_0)$ is likely to return $x_0$, 
then $A^{(\BreakOW^*,f^*)}(y_0)$ must be 
at least somewhat likely to return $x^*$, because $x^*$ has the same 
distribution as $x_0$ from $A$'s point of view
(the same argument was previously 
used in \cite{HaiHol09}, and in a more convoluted way in \cite{Simon98}).
This means that the two runs of $A$ have to differ in some call
with noticeable probability.
It is unlikely that they differ in a call to $f$, since $x^*$ was picked
at random and $A$ makes few calls to $f$.
Thus, they have to differ in some call to $\BreakOW$ with noticeable 
probability.

However, it turns out that $\BreakOW^*(w) \neq \BreakOW(w)$ for any $w$
with very low probability: it only happens in two cases.
First, if $x^*$ is the 
query which $g^{(f)}(\BreakOW(w))$ makes to $f$, but there is only one
such query, so this happens with probability $2^{-n}$ (over the choice of
$x^*$).

The other case is if there is some $v$ for which the output
of $g^{(f)}(v)$ changes to $w$ when we replace $f$ with $f^*$.

Now, recall the check $\BreakOW$ performs before it outputs $v$.
This check is equivalent to the following:
enumerate all pairs $(v',y')$, and count
the number for which $g(v') = w$ in case $f$ answers the only query with $y'$.
If this number is larger than $2^{n/30}$, refuse to return $v$.

This now implies that there can only be $2^{n/30}$ values for $x^*$
for which the output changes to $w$, and so this case is unlikely as well.

\subsection{Multiple calls}
The case when $g$ can make more than $1$ call is significantly 
more difficult than the case where $g$ makes a single call.
It turns out that most of the issues which arise can be discussed
already for $r=2$ calls, so we restrict the discussion to this case
in this section.

\paragraph{Construction with many calls}

Of course, the same examples as before still work.
Thus, $\BreakOW(w)$ still does the same 
check before returning $v$: does conditioning on one of the
two query answers $y_1$ and $y_2$ given by $f$ in the
evaluation of $f(v)$ make $w$ much more likely?
If so, it refuses to answer.

However, it turns out that we can restrict $\BreakOW(w)$ even more:
it should also not return a preimage $v$ if conditioning on 
having seen \emph{both} outputs $y_1$ and $y_2$ in an evaluation
makes the output $w$ more likely.
As it turns out, we only know how to prove that $\BreakOW$ does not
help invert $f$ with this additional restriction.

A useful example might be the construction,
which takes as input a $v = (x_1,x_2)$ of length $2n$, and is defined
as
\begin{align}
g^{(f)}(x_1,x_2) = \begin{cases}
(f(x_1),f(x_2)) & \text{if $f(x_1) = f(x_2) \oplus (1,\ldots,1)$}\\
(x_1,x_2) & \text{otherwise.}
\end{cases}
\end{align}
This will be a pseudouniform function, because 
usually $f(x_1) \neq f(x_2) \oplus (1,\ldots,1)$.
Also, we see that an adversary $A$ which tries to use $\BreakOW$
to invert $y$ would presumably call $\BreakOW$ on 
input $(y,y\oplus (1,\ldots,1))$.
However, using the additional restriction above, $\BreakOW$ will definitely
not return the inverse $A$ is looking for.

We make two additional remarks:
It turns out that if $\BreakOW$ inverts $g(v)$ with
low probability,  we
can choose $\cY \subseteq \{0,1\}^n$ as small as $2^{\Theta(n/r)}$,
and conditioned on $f$ being from $\cF(\cY)$, the output
of $g$ is very biased.
Since $\cY$ is superpolynomial only as long as $r \in o(n/\log(n))$, we
see that $f$ stops being a one-way function once $r \notin o(n/\log(n))$.

Second, there is a question on whether above one should condition
on $y_1$ being the first output, and $y_2$ being the second output,
or just on both $y_1$ and $y_2$ appearing as an output.
We choose the latter, as it seems more natural in the concentration
bound explained below.
It seems we can be relatively careless with this, 
because $r^r \ll 2^n$.
\paragraph{The underlying one-way function still remains one-way}
Again, we need to argue why $\BreakOW$ does not
help to invert $f$.
As before, we can show that we only need to prove 
that with high probability over the choice of $x^*$ we have
$\BreakOW^*(w) = \BreakOW(w)$.
Previously, this followed by a simple counting argument.
Now, it becomes more difficult.

To see why, consider
\begin{align}
g^{(f)}(x_1,x_2) = \begin{cases}
0^{2n} & \text{if $x_1 = x_2 = 1^{n}$}\\
1^{2n} & \text{otherwise, if also $x_1 = f(x_1)$ and $x_2 \neq f(x_2)$}\\
(x_1,x_2) & \text{otherwise}
\end{cases}
\end{align}
One can check that neither conditioning on a value of $y_1$, $y_2$,
or on a pair $(y_1,y_2)$ makes some output $w$ of $g$ much more likely.
Therefore, $\BreakOW(w)$ will simply return some preimage found.

Suppose now that $f$ was picked in a very unlikely way: $f(x) = x$ for any
$x$.
Then, $\BreakOW(1^{2n})$ will return $\bot$, signifying that no
preimage was found.
On the other hand, for any $x^*$ and any $f^*$ as above, 
$\BreakOW^*(1^{2n})$ will return $(x_1,x^*)$ for some $x_1$.
Thus, for \emph{some} functions $f$, $\BreakOW^*$ can behave very
differently from $\BreakOW$.

It is, however, possible to show that functions $f$ for which this happens
are very unlikely.
In case $r=2$, a usual Chernoff bound is sufficient for that.
For $r$ larger than $2$,
a concentration bound for polynomials in the style as proven by \cite{KimVu00}
seems to be needed.
We will use a bound from \cite{Holens11}, and show 
in Section~\ref{sec:concentration}
how it can be used to show that for almost all functions $f$, 
$\BreakOW^{f}(w)  \neq \BreakOW^{f^*}(w)$ has very low probability (over the
choice of $x^*$).

It turns out that this concentration bound 
breaks down if $r \in \Omega(n/\log(n))$.

\subsection{Non-uniform security reductions}
The above considerations prove Theorem~\ref{thm:main}, which exclude
constructions with uniform security proofs.
The technique given in \cite{GGKT05} allows to give security proofs
which also hold against non-uniform security proofs, and we can apply this
technique in our context.
We apply this technique in Section~\ref{sec:nubb}, giving Theorem~\ref{thm:mainB}.

\subsection{On the security parameter restriction}
Given our techniques, one might suspect that the restriction on the 
security parameter is inherent to them.
However, as we show in Section~\ref{sec:nsprc}, this is 
not the case.
Our proof will only break the
resulting pseudouniform one-way function only for infinitely many 
security parameters $k$, instead of for all but finitely many $k$
as one might hope.

This last restriction is inherent, at least as long as one
only uses underlying regular one-way functions.
The reason is that constructions exist
which do fewer than $n/\log(n)$ calls, and yield a pseudorandom
generator for infinitely many security parameters.

In order to get rid of the restriction, we use the following idea:
We consecutively find infinitely many values $\ell$ 
for which $g(\ell, \cdot)$ does fewer than $\frac{n}{\log(n)}$ queries,
where $n$ is the \emph{shortest} input length which $g$ queries on
security parameter $\ell$.
After this, we simultaneously fix $f(k,\cdot)$ for all $k$ which $g$ 
can access.
The idea is that underlying to $f(k,\cdot)$, there could be a \emph{single} 
one-way function for many different values of $k$.
Thus,  we can reduce our task to the problem
solved in the previous sections.

Of course some technical problems arise. 
These are dealt with in Section~\ref{sec:nsprc}.

\section{The Breaker Oracles}

We will give two oracles, each of which breaks one of the 
two security properties of~$g$.
The first oracle inverts $g$ with noticeable probability,
and the second oracle distinguishes the output of $g$ 
from a uniform random string.
For each security parameter $k$
we will then set $\Breaker_k$
to be one of these two oracles, 
depending on the combinatorial structure of~$g_k$.

\subsection{The inverting oracle}
The first oracle is called $\BreakOW$.
It inverts $g$ in some cases, and is given as algorithm below,
but we first explain it informally.
On input $w \in \{0,1\}^m$, $\BreakOW(w)$ 
first enumerates all possible inputs 
$v \in \{0,1\}^m$ of $g$ in lexicographic order.
For each of them it checks whether $g^{(f)}(v) = w$.
If so, it checks whether returning $v$ could help some algorithm~$A$
to invert~$f$.
For this, it calls the procedure $\SafeToAnswer$.
Roughly speaking, $\SafeToAnswer$ will 
return false in case this fixed $w$ correlates strongly with
some outputs $y \in \{0,1\}^n$ of $f$ which occured during
the evaluation of $g^{(f)}(v)$.
More exactly, $\SafeToAnswer$ enumerates all possible subsets $B$ of the
answers $f$ gave in the evaluation of $g^{(f)}(v)$.
It then computes the probability that an evaluation outputs $w$, 
conditioned on the event that the evaluation produces all outputs
in $B$.
If this probability is much larger than $2^{-m}$, $\SafeToAnswer$
will return false. 

\mbox{}\begin{mdframed}\setlength{\parindent}{0cm}%
\textbf{Algorithm $\BreakOW^{(f)}(w)$}
\medskip

\hrule

\medskip
\textbf{procedure} $\SafeToAnswer(w, Q)$: 
\centerComment{$\SafeToAnswer$ does \emph{not} depend on $f$}
\\
\mbox\qquad \textbf{for all} $B \subseteq Q$: \\
\mbox\qquad\qquad \textbf{if} $\displaystyle \Pr_{f'\leftarrow\cP,v'}[g^{(f')}(v') = w | B \subseteq \QueryY(g,v',f')] \geq 2^{-m+\frac{n}{30}}$\\
\mbox\qquad\qquad\qquad \textbf{return} false\\
\mbox\qquad \textbf{return} true\\
\textbf{done}\\
\\
\textbf{for all} $v \in \{0,1\}^m$ \textbf{do}\\
\mbox\qquad \textbf{if} $g^{(f)}(v) = w$ \textbf{then}\\
\mbox\qquad\qquad\textbf{if} $\SafeToAnswer(w, \QueryY(g,v,f))$
 \textbf{then}\\
\mbox\qquad\qquad\qquad \textbf{return} $v$\\
\textbf{return} $\bot$
\end{mdframed}

\medskip

We next define the quantity $p(g)$.
This is the probability that $\BreakOW$ 
inverts $g(v)$ by returning $v$
(actually, not quite: $\BreakOW$ might return a different
preimage of $g(v)$ before it enumerates $v$ -- in any case, 
the probability that $\BreakOW$ inverts~$g$ is at least $p(g)$).

\begin{align}\label{eq:1}
p(g) := \Pr_{{f \leftarrow \cP}\atop{v \leftarrow\{0,1\}^m}}[
\SafeToAnswer(g^{(f)}(v),\QueryY(g,f,v))]
\end{align}

It is easy to see that in case
$p(g) \geq \frac12$, then $\BreakOW^{(f)}$ will invert
$g(v)$ with noticeable probability.
\begin{lemma}\label{lem:breakOWHelpsInvertG}
Let $g^{(\cdot)}: \{0,1\}^m \to \{0,1\}^{m*}$ be a 
normalized oracle construction.
If $p(g) \geq \frac12$, then 
\begin{align}
\Pr_{f \leftarrow \cP,v}[\text{$\BreakOW(g^{(f)}(v)))$ inverts $g^{(f)}$}]
\geq \frac12.
\end{align}
\end{lemma}
\begin{proof}
Pick $v$ and $f$ at random and call $\BreakOW(w)$ for $w=g^{(f)}(v)$.
When $\BreakOW$ enumerates all possible values $v$, at one point
it will pick the actual chosen value $v$ unless it has returned
a preimage of $w$ before.  
With probability at least $\frac12$,
$\SafeToAnswer(w,Q)$ returns true where $Q = \QueryY(g,v,f)$,
in which case $\BreakOW(w)$ will return some inverse of~$w$.
\end{proof}

Our next goal is a more interesting claim: 
$\BreakOW$ is unlikely to help inverting $f$, when is uniformly drawn from $\cP$.
For this, we introduce
the following definition (which is motivated by the soon to
follow Lemma~\ref{lem:breakerStaysConstant}).

\begin{definition}\label{def:Q}
Let $g^{(\cdot)}: \{0,1\}^m \to \{0,1\}^{m*}$ be an $r$-query normalized oracle
construction.
For $f: \{0,1\}^n \to \{0,1\}^n$, $y^*\in \{0,1\}^m$, 
and $w \in \{0,1\}^m$, 
the set $Q_{f,y^*,w}$ contains all pairs $(x^*,v^*)$ with the 
following properties:
\begin{enumerate}
\item[(a)] $g^{(f^*)}(v^*) = w$
\item[(b)] $x^* \in \QueryX(g,v^*,f^{*})$, i.e., 
$g^{(f^*)}(v^*)$ queries $x^*$
\item[(c)] $\SafeToAnswer(w,\QueryY(g,v^*,f^*))$,
\end{enumerate}
where $f^{*} = f_{(x^*,y^*)}$.
\end{definition}

We will prove the next lemma in Section~\ref{sec:concentration}
(some intuition on why this is true can 
be found in Section~\ref{sec:intuition}).
It states that with very high probability over the choice of $f$,
the set $Q_{f,y^*,w}$ is small.

\begin{lemma}\label{lem:concentration}
Let $g^{(\cdot)}: \{0,1\}^m \to \{0,1\}^{m*}$ be an $r$-query normalized oracle
construction, $r \leq \frac{n}{100 \log(n)}$.
For all ($w, y^*)$ we have
\begin{align}
\Pr_{f \leftarrow \cP}\bigl[|Q_{f,w,y^*}| > 2^{\frac{n}{10}}\bigr] 
< 2^{-2^{\frac{n}{100r}}}\;.
\end{align}
\end{lemma}

Fix now some permutation $f$, some $y^* \in \{0,1\}^n$ and some
$w \in \{0,1\}^m$.
Compare runs of $\BreakOW^{(f)}(w)$ and 
$\BreakOW^{(f_{(x^*,y^*)})}(w)$ for a random element 
$x^* \in \{0,1\}^n$.
The next lemma shows that the result of these two runs is
equal with high probability in case $|Q_{f,y^*,w}|$ is small.

\begin{lemma}\label{lem:breakerStaysConstant}
Fix $f$, $y^*$, $w$.
If $|Q_{f,y^*,w}| \leq 2^{\frac{n}{10}}$, then
\begin{align}
\Pr_{x^*}[\BreakOW^{(f)}(w) \neq \BreakOW^{(f^*)}(w)]
\leq
2^{-\frac{4n}{5}}\,
\end{align}
where $f^* = f_{(x^*,y^*)}$.
\end{lemma}
\begin{proof}
Let $v$ be the result of $\BreakOW^{(f)}(w)$, 
and $v^*$ the result of $\BreakOW^{(f^*)}(w)$.
We distinguish two cases.

First, suppose that $v^* = \bot$ or that 
$v^*$ occurs in the enumeration of $\BreakOW$
after $v$.
This can only happen if $x^* \in \QueryX(g,v,f)$,
because if not, $\BreakOW^{(f^*)}(w)$ will behave 
exactly the same in the iteration of $v$, and so it must
also return $v$.

Second, suppose that $v = \bot$ or that
$v$ occurs in the enumeration of $\BreakOW$
after $v^*$.
We claim that in this case $(x^*,v^*) \in Q_{f,y^*,w}$.
Clearly, conditions (a) and (c) in Definition \ref{def:Q}
must hold, as otherwise $\BreakOW^{(f^*)}(w)$ will not output 
$v^*$.
Condition (b) must also hold.
Otherwise we have that $g^{(f)}(v^*) = w$
(because of (a) and the fact that $x^*$ has not been queried)
and $\QueryY(g,v^*,f) = \QueryY(g,v^*,f^*)$.
This would imply that
$\SafeToAnswer(w,\QueryY(g,v^*,f)) = 
\SafeToAnswer(w,\QueryY(g,v^*,f^*))$,
and so we see that if (b) would not hold, $\BreakOW^{(f)}(w) = v^*$.

Since the union of the sets $\QueryX(g,v,f)$
and $Q_{f,y^*,w}$ has fewer than $2^{\frac{n}{5}}$
elements the result follows.
\end{proof}

Now we can show that $\BreakOW$ usually does not help 
to invert $f$.
\begin{lemma}\label{lem:breakerHelpsNothing}
Let $g^{(\cdot)}: \{0,1\}^m \to \{0,1\}^{m*}$ be an $r$-query 
normalized oracle construction, $r < \frac{n}{100\log(2n+m)}$.
Let $A^{f,\BreakOW}$ be an arbitrary algorithm making at most $2^{\frac{n}{20}}$
queries to $f$ and to $\BreakOW$.
Then, the probability that $A$ inverts $f(x)$ 
is at most
\begin{align}
\Pr_{f\leftarrow\cP,x,A}
[\text{$A^{f,\BreakOW}(f(x))$ inverts $f$}] \leq 
2^{-\frac{n}{30}}\;.
\end{align}
\end{lemma}

\begin{proof}
First, because $f$ is picked from the set of permutations $\cP$, 
we see that
\begin{align}
\Pr_{x,f\leftarrow\cP,A}[\text{$A^{(\BreakOW,f)}(f(x))$ inverts $f(x)$}] 
=
\Pr_{x,f\leftarrow\cP,A}[A^{(\BreakOW,f)}(f(x)) =x]
\end{align}
Fix now an arbitrary function $f$.
In case $f$ is such that for all pairs $(w,y^*)$ 
the bound $|Q_{f,w,y^*}| \leq 2^{\frac{n}{10}}$ holds, we get
for any~$x$ and any fixed randomness of~$A$ 
\begin{align}
\Pr_{x^*}[A^{(\BreakOW,f)}(f(x)) \neq A^{(\BreakOW^*,f^*)}(f(x))] \leq 
2^{\frac{n}{20}} 2^{-\frac{4n}{5}}
< 2^{-\frac{n}{20}}
\end{align}
where $f^* = f_{(x^*,f(x))}$, $\BreakOW = \BreakOW^{(f)}$,
and $\BreakOW^* = \BreakOW^{(f^*)}$.
This holds because any of the $2^{\frac{n}{20}}$ calls to either oracle
will return the same answer with probability $2^{-\frac{4n}{5}}$ 
(using Lemma~\ref{lem:breakerStaysConstant} for calls to $\BreakOW$,
for calls to $f$ this is obvious).

We can also pick $x$ and $f$ at random, then we get
\begin{align}
&\Pr_{f,x,x^*}[A^{(\BreakOW,f)}(f(x)) \neq A^{(\BreakOW^*,f^*)}(f(x))] \nonumber\\
&\qquad \leq \Pr_{f}[\exists (w,y^*): |Q_{f,w,y^*}| > 2^{\frac{n}{10}}] + 
2^{-\frac{n}{20}}\nonumber\\
&\qquad \leq 2^{m+n-2^{\frac{n}{100r}}} + 2^{-\frac{n}{20}} < 
2^{-\frac{n}{20}+1}\;,\label{eq:13}
\end{align}
where we applied Lemma~\ref{lem:concentration}.

Still fixing the randomness of $A$, we see that
\begin{align}
\Pr_{f,x}[A^{(\BreakOW,f)}(f(x)) =x]
&\leq
\Pr_{f,x^*,x}[A^{(\BreakOW^*,f^*)}(f(x)) =x] + 2^{-\frac{n}{20}+1}\label{eq:21}\\
&=
\Pr_{f,x^*,x}[A^{(\BreakOW^*,f^*)}(f(x)) =x^*] + 2^{-\frac{n}{20}+1}\label{eq:11}\\
&\leq
\Pr_{f,x^*,x}[A^{(\BreakOW,f)}(f(x)) =x^*] + 2^{-\frac{n}{20}+2}\label{eq:12}\\
&=
2^{-n}+2^{-\frac{n}{20}+2}< 2^{-\frac{n}{30}}\;,
\end{align}
where we get (\ref{eq:11}) because the triples $(f^*, f(x), x)$ and
$(f^*,f(x),x^*)$ have exactly the same distribution.
We used (\ref{eq:13}) to get (\ref{eq:21}) and (\ref{eq:12}).

Since this holds for each random choice $A$ can make, 
it must also hold overall.
\end{proof}

\subsection{The distinguishing oracle}

Oracle $\BreakOW$ described above works well in case
$p(g) \geq \frac12$.
Therefore, we now concentrate on the case $p(g) \leq \frac12$.
In this case, there are elements $y_1,\ldots,y_b$ such that
conditioned on those occuring as outputs of $f$,
some elements $w$ are much more likely than others 
(in fact, on a random evaluation we have probability at least
 $\frac12$ that a subset of the $y$'s produced satisfies this).
Thus, it is not too far fetched to hope that if 
$f$ is a function $f: \{0,1\}^n \to \cY$ for 
some set $\cY \subseteq \{0,1\}^n$ which is small, 
then often $g^{(f)}(v)$ will
be one of few possible values.  
Formally, we can prove the following lemma.
\begin{lemma}\label{lem:largeP}
Let $g^{(\cdot)}: \{0,1\}^m \to \{0,1\}^{m*}$ be an 
$r$-query normalized oracle
construction with $p(g) \leq \frac12$, $\frac{n}{1000r} \in \mathbb{N}$.
There exists $\cY \subseteq \{0,1\}^n$
of size $|\cY| = 2^{\frac{n}{100r}}$ 
and a set $W \subseteq \{0,1\}^m$
of size $|W| \leq 2^{m - \frac{n}{100}}$ such that 
\begin{align}\label{eq:16}
\Pr_{{f \leftarrow \cF(\cY)}\atop{ v\leftarrow\{0,1\}^m}}
[g^{(f)}(v) \in W] \geq \frac{1}{2} - r^2 2^{-\frac{n}{100r}}
\end{align}
\end{lemma}

\begin{proof}
We pick $\cY \subseteq \{0,1\}^n$ of size $2^{\frac{n}{100r}}$ uniformly
at random, and then set
\begin{align}
W = \Bigl\{w \Bigm| \exists Q\subseteq \cY: |Q| = r \land \lnot 
\SafeToAnswer(w,Q)
\Bigr\}\;.
\end{align}

We start by showing that $|W| \leq 2^{m-\frac{n}{100}}$.
There are
fewer than $(|Y|)^{r} = 2^{\frac{n}{100}}$ subsets $Q \subseteq \cY$
of size $|Q| = r$, and for each of them,
$\SafeToAnswer$ considers $2^{r} \leq 2^{\frac{n}{100}}$ subsets~$B$.
For each $B$, there can be at most $2^{m-\frac{n}{30}}$ elements~$w$
which have probability at least $2^{-m+\frac{n}{30}}$ conditioned
on $B \subseteq \QueryY(g,v,f')$.
Thus, in total there can be at most $2^{m-\frac{n}{30}+\frac{n}{100}+\frac{n}{100}} \leq 2^{m-\frac{n}{100}}$
elements in $W$.

To see (\ref{eq:16}), we note first that
\begin{align}\label{eq:9}
\Pr_{{\cY,f \leftarrow \cF(\cY)}\atop{ v\leftarrow\{0,1\}^m}}
[g^{(f)}(v) \in W]
=
\Pr_{{\cY, v\leftarrow\{0,1\}^m}\atop{(y_1,\ldots,y_r) \leftarrow (P(r,\cY))}}
[g'(v,y_1,\ldots,y_r) \in W]\;,
\end{align}
where the distribution $P(r,\cY)$ over $\cY^{r}$ is
the distribution of $(f(x_0),\ldots,f(x_{r-1}))$ for some fixed
pairwise disjoint values $x_0,\ldots,x_{r-1}$ and $f \leftarrow \cF(\cY)$.
It has the following two properties.
First, the probability that $P(r,\cY)$ gives $r$ pairwise
disjoint outputs is at least $1 - \frac{r^2}{|\cY|}$ (by a union bound).
Second, all tuples $(y_1,\ldots,y_r)$ in which the elements are
pairwise disjoint have the same probability when the probability
is also over the choice of $\cY$.

Thus,
\begin{align}
\!\!\!\!\!\!\!\!\Pr_{{\cY, v\leftarrow\{0,1\}^m}\atop{(y_1,\ldots,y_r) \leftarrow (P(r,\cY))}}
\!\!\!\!\!\!\!\!&[g'(v,y_1,\ldots,y_r) \in W]\\
&\geq
\Pr_{{\cY, v\leftarrow\{0,1\}^m}\atop{(y_1,\ldots,y_r) \leftarrow (P(r,\cY))}}
\bigl[|\{y_1,\ldots,y_r\}| = r\bigr] \times \\
&\qquad\qquad
\Pr_{{\cY, v\leftarrow\{0,1\}^m}\atop{(y_1,\ldots,y_r) \leftarrow (P(r,\cY))}}
\bigl[g'(v,y_1,\ldots,y_r) \in W \bigm| |\{y_1,\ldots,y_r\}| = r\bigr]\\
&\geq
\Pr_{{\cY, v\leftarrow\{0,1\}^m}\atop{(y_1,\ldots,y_r) \leftarrow (P(r,\cY))}}
\bigl[|\{y_1,\ldots,y_r\}| = r\bigr] \times \\
&\qquad\qquad
\Pr_{(y_1,\ldots,y_r)}
[g'(v,y_1,\ldots,y_r) \in W]\;,
\end{align}
where in this last probability the values $y_1,\ldots,y_r$ are picked 
uniformly without repetition.
Next, we see that
\begin{align}
 \Pr_{(y_1,\ldots,y_r)}
&\bigl[g'(v,y_1,\ldots,y_r) \in W \bigr]
\\
&\qquad \geq
 \Pr_{(y_1,\ldots,y_r)}
\bigl[\lnot \SafeToAnswer(g'(v,y_1,\ldots,y_r), \{y_1,\ldots,y_r\})]
\label{eq:23}\\
&\qquad = 1-p(g)
\end{align}
because without repetition the $y_i$  have exactly the same 
distribution as in the definition of $p(g)$.
In total,
\begin{align}\label{eq:p9}
\Pr_{{\cY,f \leftarrow \cF(\cY)}\atop{ v\leftarrow\{0,1\}^m}}
[g^{(f)}(v) \in W]
\geq
\Bigl(1-\frac{r^2}{|\cY|}\Bigr)(1- p(g))\;.
\end{align}
\end{proof}

Let now $\BreakPU(W)$ be the oracle which on input $w$ returns
$1$ if and only if $w \in W$.
The next lemma states that $\BreakPU(W)$ does not 
help significantly in inverting $f$.
This is intuitive, since it does not even depend on $f$ 
(besides the choice of~$\cY$).
Furthermore, this lemma also follows directly 
from \cite[Theorem 1]{GGKT05}.
To see this, note that we can pick
$f$ as follows: first pick any regular function 
$p: \{0,1\}^n \to \cY$ and then set $f = \pi \circ p$ for some
permutation $\pi$; by \cite[Theorem 1]{GGKT05}, $f$ is 
$2^{|\cY|^{(1/5)}}$-hard to invert even given $p$.
We provide a proof anyhow for completeness.
\begin{lemma}\label{lem:breakPUhelpsNothing}
Let $A$ be an arbitrary oracle algorithm making at most
$2^{\frac{n}{1000r}}$ queries, $|\cY| = 2^{\frac{n}{100r}}$, 
$\frac{n}{1000r} \in \bbN$.
Then,
\begin{align}
\Pr_{f\leftarrow\cF(\cY),x,A}
[\text{$A^{f,\BreakPU}(f(x))$ inverts $f$}] \leq 
2^{-\frac{n}{1000r}}
\;,
\end{align}
where $\BreakPU = \BreakPU(W)$ for an arbitrary set $W$.
\end{lemma}
The proof is similar to the proof of Lemma~\ref{lem:breakerHelpsNothing}.
\begin{proof}
We first note that
\begin{align}
\Pr_{{f \leftarrow \cF(\cY)}\atop{x,A}}
[\text{$A^{f,\BreakPU}(f(x))$ inverts $f(x)$}]
=
\frac{2^n}{|\cY|}
\Pr_{{f \leftarrow \cF(\cY)}\atop{x,A}}
[A^{f,\BreakPU}(f(x)) = x]
\end{align}
because for any fixed $f(x)$, the value of $x$ is still uniform among
the $\frac{2^n}{|\cY|}$ preimages.

Now, fix any $x$ and any $f$, and let $q = 2^{\frac{n}{1000r}}$ 
be the upper bound on the number of queries by $A$.
Keeping the randomness of $A$ fixed, 
\begin{align}\label{eq:10}
\Pr_{x^*}
[A^{f,\BreakPU}(f(x)) \neq A^{f^*,\BreakPU}(f(x))] \leq \frac{q}{2^n}\;,
\end{align}
where $f^* = f_{(x^*,f(x))}$, because the output
of $A^{f^*,\BreakPU}(f(x))$ can only differ from 
$A^{f,\BreakPU}(f(x))$ in case $x^*$
is one of the elements on which $A$ queried $f$.

As in the proof of Lemma~\ref{lem:breakerHelpsNothing},
\begin{align}
\Pr_{{f\leftarrow\cF(\cY)}\atop{x,x^*}}[A^{f,\BreakPU}(f(x)) = x] &\leq
\Pr_{{f\leftarrow\cF(\cY)}\atop{x,x^*}}[A^{f^*,\BreakPU}(f(x)) = x] + \frac{q}{2^n}\\
&=
\Pr_{{f\leftarrow\cF(\cY)}\atop{x,x^*}}[A^{f^*,\BreakPU}(f(x)) = x^*] + \frac{q}{2^n}
\leq
\frac{2q+1}{2^n}\;.
\end{align}

Together,
\begin{align}
\Pr_{x\leftarrow \{0,1\}^{n},f}
[\text{$A^{f,\BreakPU}(f(x))$ inverts $f(x)$}]
\leq
\frac{2^n}{|\cY|} \frac{2q+1}{2^n} = \frac{2q+1}{|\cY|}\;.
\end{align}
\end{proof}

\subsection{Proving the main result}

The above lemmas can be used to prove Theorem~\ref{thm:main},
which we restate here for reference.

\maintheorem*

\begin{proof}
In order to get a contradiction, we assume otherwise.
Because of Theorem~\ref{thm:prgandPUOWF}, we can also
assume that we have a fully black-box reduction which gives a 
pseudouniform one-way function (which is defined in a way
analogous to Definition~\ref{def:fbb}).

Thus, suppose we have some construction $(g,A)$.
We we want to instantiate the construction
with length preserving one-way functions, where the
input and output length equals the security parameter
$k$, i.e., $n(k) := n'(k) := k$.
The construction must work for this choice by definition.

We can assume that $\frac{n(k)}{1000r(k)} \in \bbN$ 
for all but finitely many $k$, because we can increase $r(k)$ 
such that this holds and such that still $r(k) \in o(\frac{n(k)}{\log(n(k))})$.

We now make sure that our construction is normalized.
For this, we modify $g$ such that it makes exactly $r(k)$ pairwise
disjoint queries to $f$; clearly, this is no problem.

We then define
\begin{align}
{\tilde g}^{(f_k)}(v) :=
\begin{cases}
(\bot,v) & \text{if two queries of $f_k$ yield the same output}\\
g_k^{(f_k)}(v) & \text{otherwise.}
\end{cases}
\end{align}

Next, we will provide, for each $k$ seperately, two oracles
$f$ and $B = \Breaker$. We construct these oracles such that $B$ breaks the security
property of $g_k$ for all but finitely many $k$, and yet
the probability that $A^{\Breaker,f}$ inverts $f$ is negligible.

For this, we consider $p(\tilde{g}_k)$ for each $k$ seperately.
If $p(\tilde{g}_k) \geq \frac12$ we 
set $\Breaker_k$ to be $\BreakOW$.
By Lemma~\ref{lem:breakOWHelpsInvertG} we see 
that for these $k$ 
\begin{align}
\Pr_{f_k \leftarrow \cP_k,v}[\text{$\Breaker_k(\tilde{g}_k^{(f_k)}(v)))$ inverts $\tilde{g}_k^{(f_k)}$}]
\geq \frac12.
\end{align}
By Lemma~\ref{lem:breakerHelpsNothing} we also see that, if $k$ is large enough,
\begin{align}
\Pr_{f_k\leftarrow\cP_k,x,A}
[\text{$A^{f_k,\Breaker_k}(f_k(x))$ inverts $f_k$}] \leq 
2^{-\frac{n(k)}{30}}\;.
\end{align}
Note that in this case, $\tilde{g}_k$ behaves the same as $g_k$, because
no two queries to $f_k$ can output the same value.
Applying Markov's inequality, for fraction at least $\frac{1}{10}$
of the functions $f_k$ we have
\begin{align}\label{eq:19}
\Pr_{v}[\text{$\Breaker_k(g_k^{(f_k)}(v)))$ inverts $g_k^{(f_k)}$}]
&\geq \frac1{10}
\end{align}
Furthermore, for fraction at least $\frac{99}{100}$ 
of the functions $f_k$ we have
\begin{align}\label{eq:20}
\Pr_{x,A}
[\text{$A^{f_k,\Breaker_k}(f_k(x))$ inverts $f_k$}] &\leq 
100 \cdot 2^{-\frac{n(k)}{30}}\;.
\end{align}
We pick a function $f_k$ for which both (\ref{eq:19}) and (\ref{eq:20})
are satisfied.

If
 $p(\tilde{g}_k) \leq \frac12$, Lemma~\ref{lem:largeP}
gives a set  $W_k \subseteq \{0,1\}^{m(k)}$ and 
$\cY_k \subseteq \{0,1\}^{n(k)}$.
For $n$ large enough, $A$ satisfies
the requirements of Lemma~\ref{lem:breakPUhelpsNothing}, and we see that
\begin{align}\label{eq:4}
\Pr_{{f_k \leftarrow \cF(\cY_k)}\atop{A,x \leftarrow \{0,1\}^{n(k)}}}
[\text{$A^{\BreakPU(W_k),f_k}(f_k(x))$ inverts $f_k(x)$}] 
\leq 2^{-\frac{n}{1000r}},
\end{align}
which is negligible.
By Lemma~\ref{lem:largeP}
\begin{align}\label{eq:7}
\Pr_{{f_k \leftarrow \cF(\cY_k)}\atop{v\leftarrow\{0,1\}^{m(k)}}}
[\tilde{g}_k^{(f_k)}(v) \in W_k] \geq \frac{1}{4},
\end{align}
again for $k$ is large enough.
Because $\cY$ is of superpolynomial size, the probability
that $\tilde{g}_k$ outputs $(\bot,v)$ is still negligible.
Thus, we can argue
as before, and there is some choice of $f_k$ for which 
\begin{align}
\Pr_{{A,x \leftarrow \{0,1\}^{n(k)}}}
[\text{$A^{\Breaker_k,f_k}(f_k(x))$ inverts $f_k(x)$}] 
&\leq 2^{-\frac{n}{110r}}\qquad\text{, and}\\
\Pr_{{v\leftarrow\{0,1\}^{m(k)}}}
[g_k^{(f_k)}(v) \in W_k] &\geq \frac{1}{10}.
\end{align}
We fix such a choice of for $f_k$ and set $\Breaker_k := \BreakPU(W_k)$.

We conclude that while the statement analogous to (\ref{eq:1000}) holds
(for breaking the either the pseudouniformity or for inverting $g$), 
the statement (\ref{eq:1001}) fails to hold, and so we get a contradiction.
\end{proof}

\section{Proof of Lemma~\ref{lem:concentration}}
\label{sec:concentration}
In this section, we give the proof of Lemma~\ref{lem:concentration}.
However, before giving the proof, we provide some intuition
in Section~\ref{sec:intuition} (which can be skipped if desired).
\subsection{Intuition}\label{sec:intuition}
Fix $(f,y^*,w)$, and assume that $(x^*,v^*) \in Q_{f,y^*,w}$.
Consider the query-answer pairs
$\{(x_1,y_1),\ldots,(x_r,y_r)\} = \Query(g,
v^*, f_{(x^*,y^*)})$ which occur in an evaluation of 
$g^{(f_{(x^*,y^*)})}(v^*)$.
The pair $(x^*,y^*)$ must be in this set, as otherwise conditions
(a) or (b) of Definition~\ref{def:Q} would not hold,
and to simplify the discussion we make the (unrealistic)
assumption that
always $(x^*,y^*) = (x_r,y_r)$.
Now consider the set $T = \{(x_1,y_1),\ldots,(x_{r-1},y_{r-1})\}$.
Let us call $T$ an \emph{incrementor} for $|Q_{f,y^*,w}|$,
because whenever $f$ satisfies $f(x_i) = y_i$ for $i \in \{1,\ldots,r-1\}$,
the set $Q_{f,y^*,w}$ grows by $1$.\footnote{Ignoring a few
reasons why this might not be true sometimes\ldots like the
fact that $\SafeToAnswer$ might return false.}

Now, still fixing $(f,y^*,w)$, 
the total number of such ``incrementors''
for $|Q_{f,y^*,w}|$
is at most $2^{(r-1)n + \frac{n}{30}}$.
To see this, we argue that otherwise, 
(for $y_r$ being the answer of the $r$-th query 
in the evaluation)
\begin{align}\label{eq:14}
\Pr_{f'\leftarrow \cP,v'}[g^{f'}(v') = w| y_r = y^*] 
\geq 2^{-m+\frac{n}{30}}\;,
\end{align}
because any of the incrementors
survive\footnote{Formally, surviving means that 
$f(x_i) = y_i$ for all pairs $(x_i,y_i)$
in the incrementor.} the picking 
of $f$ with probability roughly\footnote{Ignoring
very slight dependence in this discussion which arises from the fact that
$f$ is picked as a permutation.} $2^{-(r-1)n}$.
Thus, if there are $2^{(r-1)n + \frac{n}{30}}$ incrementors, in expectation  
$2^{\frac{n}{30}}$ will survive the picking of $f$, 
and if we pick one\footnote{Only one incrementor with 
a fixed $v^*$ can survive with our assumptions.}
of the $2^{\frac{n}{30}}$ values $v^*$ which survived
we get an element for which $g^{f'}(v') = w$ (conditioning on $y_r = y^*$).
Now, (\ref{eq:14}) roughly contradicts
$\SafeToAnswer(w,Q)$ for $B = \{y^*\}$ (up to some issues due
to our simplifying assumption that $(x^*,y^*)$ is always $(x_r,y_r)$,
but since $r^r < 2^n$ they do not matter much).

Thus,  there are at most $2^{(r-1)n + \frac{n}{30}}$ incrementors
for $|Q_{f,y^*,w}|$, and so in expectation $|Q_{f,w,y^*}| \leq 2^{\frac{n}{30}}$.
However, we need to prove that the $|Q_{f,w,y^*}|$ is small with
(very) high probability, and not in expectation.
Luckily for us, Kim and Vu \cite{KimVu00} proved a concentration
bound which can be applied in our setting -- translated to our
setting, they show that concentration \emph{does} hold if several conditions
are given.
First, it needs to hold that all probabilities 
checked in $\SafeToAnswer$ are smaller than $2^{-m+\frac{n}{30}}$
(which is, besides Lemma~\ref{lem:largeP},
the reason that $\SafeToAnswer$ is defined 
in the way it is defined).
Second, they roughly require that $r^r < 2^n$, which holds
in our case, because we assume that $r \notin \Omega(\frac{n}{\log(n)})$.
Finally, they require that the events $f(x_1) = y_1$
and $f(x_2) = y_2$ are independent---which of course is a problem,
because this does not hold in our case.
Luckily, it turns out that this last requirement can be relaxed somewhat
using a proof technique implicit in \cite{ScSiSr95} (see
\cite{Rao08, ImpKab10}).
A proof of a Kim-Vu style concentration bound in this form
was given by the first author in \cite{Holens11}.

\subsection{The polynomial $P_{w,y^*}$}
To prove Lemma~\ref{lem:concentration}, 
we will first find a polynomial $P_{w,y^*}$ of degree $r$ in 
variables $F_{(x,y)}$ for all $x,y \in \{0,1\}^n$.
The polynomial will have the following property:
fix an arbitrary function $f: \{0,1\}^n \to \{0,1\}$, and
set the variables $F_{(x,y)}$ as follows:
\begin{align}\label{eq:5}
F_{(x,y)} = \begin{cases}
1 & \text{if $f(x) = y$} \\
0 & \text{otherwise.}
\end{cases}
\end{align}
We will see that the value of $P_{w,y^*}$ for these values (evaluated 
over $\mathbb{R}$ or $\mathbb{N}$) gives an upper bound on $|Q_{f,y^*,w}|$.
We denote this value by $P_{w,y^*}(f)$.

The polynomial $P_{w,y^*}$ is obtained by a run of
 algorithm $\BuildPolynomial(w,y^*)$.

\mbox{}\begin{mdframed}\setlength{\parindent}{0cm}%
\textbf{Algorithm} $\BuildPolynomial(w,y^*)$
\medskip

\hrule

\medskip

$P_{w,y^*} := 0$\\
 \textbf{forall} $(y_1,\ldots,y_r) \in (\{0,1\}^{n})^r$ \textbf{do}\\
\mbox\qquad \textbf{if} $\SafeToAnswer(w,\{y_1,\ldots,y_r\})$ \textbf{then}\\
\mbox\qquad\qquad\textbf{forall} $v \in \{0,1\}^m$ \textbf{do}\\
\mbox\qquad\qquad\qquad \textbf{if} $g'(v,y_1,\ldots,y_r) = w$ $\land$ $y^* \in \{y_1,\ldots,y_r\}$ \textbf{then}\\
\mbox\qquad\qquad\qquad\qquad $T := \{(x_i,y_i): 
\text{$i \in \{1,\ldots,r\}$ and $y_i \neq y^*$ and \ldots}$\\
\mbox\qquad\qquad\qquad\qquad\qquad\qquad $\text{$x_i$ is the $i$th query done by $g'(v,y_1,\ldots,y_r)$}\}
$\\
\mbox\qquad\qquad\qquad\qquad $\displaystyle P_{w,y^*} := P_{w,y^*} + \prod_{(x,y) \in T} F_{(x,y)}$\\
\textbf{return} $P_{w,y^*}$
\end{mdframed}

\medskip

For readers who did not skip the intuition, we can
connect this with Section~\ref{sec:intuition}.
The term $\prod_{(x,y)} F_{(x,y)}$ corresponds to an incrementor,
and we note that if the incrementor survives the picking, the summand
in the polynomial will evaluate to~$1$.

\begin{lemma}\label{lem:polynomialGivesABound}
For any $f$, $w$, $y^*$ we have $|Q_{f,w,y^*}| \leq P_{w,y^*}(f)$.
\end{lemma}
\begin{proof}
Pick $f$ at first, and then consider a run of BuildPolynomial.
We show that for each pair $(x^*,v^*) \in Q_{f,w,y^*}$
the procedure $\BuildPolynomial(w,y^*)$ adds a monomial to 
$P_{w,y^*}$ which evaluates to $1$ under $f$.

Fix now a pair $(x^*,v^*) \in Q_{f,w,y^*}$, and let
$(x_1,y_1),\ldots,(x_r,y_r)$ be the pairs of queries
and answers made to $f_{(x^*,y^*)}$ in an evaluation of 
$g^{(f_{(x^*,y^*)})}(v^*)$.
It must be that $x^* \in \{x_1,\ldots,x_r\}$, because 
of condition (b) in Definition~\ref{def:Q}, and so $(x^*,y^*) \in 
\{(x_1,y_1),\ldots,(x_{r},y_r)\}$.
$\SafeToAnswer(w,\{y_1,\ldots,y_r\})$ must also hold (as otherwise
$(x^*,v^*) \notin Q_{f,w,y^*}$).
Thus, when $\BuildPolynomial$ enumerates the values $(y_1,\ldots,y_r)$
and $v^*$, it adds
$\prod_{(x,y) \in T} F_{(x,y)}$ to $P_{w,y^*}$, which is $1$ for the
assignment given by $f$ to the variables
 (note that $T$ does not contain $(x^*,y^*)$).
\end{proof}

\subsection{Derivatives of $P_{w,y^*}$}
Let now $B \subset \{ F_{(x,y)} \}$ be a subset of the random
variables $F_{(x,y)}$.
For any multilinear polynomial $P$ in the variables $\{F_{(x,y)}\}$ we let
$\del_{B}P$ be the formal derivative of $P$ with 
respect to the variables in $B$.
For example,
$\del_{\{F_{(1,1)},F_{(2,2)}\}}(F_{(1,1)}F_{(2,2)}F_{(3,3)}
 + F_{(1,1)}F_{(3,3)}F_{(4,4)}) = F_{(3,3)}$.)

Let $\cF^*$ be the distribution over the variables $F_{(x,y)}$ in 
which each $F_{(x,y)}$ is $1$ with probability $\frac{1}{2^n}$ 
and $0$ otherwise, and \emph{all variables are independent}.
When we pick the variables according to this distribution, they
usually cannot have been derived from a function $f$ as in 
(\ref{eq:5}).
Nevertheless, this distribution is useful to express combinatorial
properties of our polynomials.
We denote the value of the polynomial evaluated at such a 
point $\mathbf{F}$ by $P_{w,y^*}(\mathbf{F})$.

\begin{lemma}\label{lem:derivativeLow}
For any $B \subseteq (\{0,1\}^{n})^2$ and any $(w,y^*)$:
\begin{align}\label{eq:2}
\E_{\mathbf{F}\leftarrow\cF^*}\bigl[(\del_{B}P_{w,y^*})(\mathbf{F})\bigr] 
\leq 2^{\frac{2n}{30}}.
\end{align}
\end{lemma}
\begin{proof}
Suppose otherwise, and fix a triple $(B,w,y^*)$ 
for which (\ref{eq:2}) fails to hold.
We will derive a contradiction.

The polynomial $\del_{B}P_{w,y^*}$ is the sum of all monomials
in $P_{w,y^*}$ which contain the factor $\prod_{(x,y) \in B}F_{(x,y)}$,
but with this factor removed.
Each such summand contributes $2^{-n (r-1-|B|)}$ to the expectation
in~(\ref{eq:2}),
and so there are at least $2^{\frac{2n}{30} + n(r-1-|B|)}$ monomials
containing a factor $\prod_{(x,y) \in B}F_{(x,y)}$ in $P_{w,y^*}$.
This implies that there are at least that many monomials containing 
a factor of the form $\prod_{(x,y) \in B}F_{(*,y)}$, where $F_{(*,y)}$ 
is an arbitrary variable $F_{(x',y')}$ with $y' = y$.

The following algorithm
adds $2^{-n(r-1-|B|)}$ to a counter for each 
such monomial in $P_{w,y^*}$.
Therefore, it outputs at least $2^{\frac{2n}{30}}$.

\mbox{}\begin{mdframed}\setlength{\parindent}{0cm}%
\noindent \textbf{Algorithm} UpperBound$(B,w,y^*)$
\medskip

\hrule

\medskip

$E_{B,w,y^*} := 0$\\
$B' := \{y^*\} \cup \{y : (x, y) \in B\}$\\
 \textbf{forall} $(y_1,\ldots,y_r) \in (\{0,1\}^{n})^r$ \textbf{do}\\
\mbox\qquad \textbf{if} $B' \subseteq \{y_1,\ldots,y_r\}$ \textbf{then}\\
\mbox\qquad\qquad \textbf{if} $\SafeToAnswer(w,\{y_1,\ldots,y_r\})$ \textbf{then}\\
\mbox\qquad\qquad\qquad\textbf{forall} $v \in \{0,1\}^m$ \textbf{do}\\
\mbox\qquad\qquad\qquad\qquad \textbf{if} $g'(v,y_1,\ldots,y_r) = w$ \textbf{then}\\
\mbox\qquad\qquad\qquad\qquad\qquad $E_{B,w,y^*} := E_{B,w,y^*} + 2^{-n(r-|B'|)}$\\
\textbf{return} $E_{B,w,y^*}$
\end{mdframed}

\medskip

Consider now an arbitrary $(y_1,\ldots,y_r)$ 
for which $E_{B,w,y^*}$ gets increased in this algorithm.
We want to show that $\lnot\SafeToAnswer(w, \{y_1,\ldots,y_r\})$,
i.e., we want to show that
\begin{align}\label{eq:3} 
\Pr_{f',v'}[g^{(f')}(v') = w | B'' \subseteq \QueryY(g,v',f')] \geq 
2^{-m + \frac{n}{30}}
\end{align}
for some $B'' \subseteq \{y_1,\ldots,y_r\}$.
Of course, it suffices to show this for $B'' = B'$.

To see that (\ref{eq:3}) holds for $B'' = B'$ we compute
\begin{align}
\Pr_{f',v'}[&g^{(f')}(v') = w | B' \subseteq \QueryY(g,v',f')] \\
&=
\frac{\Pr_{f',v'}[g^{(f')}(v') = w \land B' \subseteq \QueryY(g,v',f')]}
{\Pr_{f',v'}[B' \subseteq \QueryY(g,v',f')]}\\
&=\frac{\Pr_{v',y_1,\ldots,y_r}[g'(v',y_1,\ldots,y_r) = w \land B' \subseteq \{y_1,\ldots,y_r\}]}
{\Pr_{y_1,\ldots,y_r}[B' \subseteq \{y_1,\ldots,y_r\}]}\label{eq:15}
\\&\geq
\frac{2^{(r-|B'|)n + \frac{2n}{30}} 2^{-m - rn}}
{\binom{r}{|B'|}(|B'|!)2^{-(|B'|)(n-1)}}\label{eq:8}
\\&\geq
\frac{2^{-m - |B'|n + \frac{2n}{30}}}{2^{r\log(r)} 2^{-|B'|(n-1)}} 
= 2^{-m + \frac{2n}{30} - r\log(r)-r} \geq 2^{-m + \frac{n}{30}}\;,
\end{align}
where in (\ref{eq:15}) and afterwards, $y_1,\ldots,y_r$
are picked uniformly from $\{0,1\}^n$, but without repetition.
The numerator in (\ref{eq:8}) can then be seen as follows:
first, note that the probability only decreases if one picks the $y_i$
with repetition, but additionally requires them to be different
for the event to occur.  
After that, one notices
that there must be at least $2^{(r-|B'|)n + \frac{2n}{30}}$ 
tuples $(v,y_1,\ldots,y_r)$ for
which $E_{B,w,y^*}$ gets increased in the algorithm UpperBound.
The denominator follows by noting that we can first choose 
how to make the assignment of the values in $B'$ to the elements
$(y_1,\ldots,y_r)$ (there are $\binom{r}{|B'|}(|B'|!)$ possibilities
for this), and then checking whether this assignment
occurs, which happens with probability at most $2^{-|B'|(n-1)}.$
Thus, we get that $\lnot\SafeToAnswer(w,\{y_1,\ldots,y_r\})$ must
hold for any tuple where $E_{B,w,y^*}$ is increased, which is the 
required contradiction.
\end{proof}

\subsection{Kim-Vu style concentration}

In a fundamental paper \cite{KimVu00}, Kim and Vu consider low degree 
polynomials $P$ in variables $x_1,\ldots,x_\ell$, and show that 
if $\del_B P$ can be bounded 
(as in Lemma~\ref{lem:derivativeLow}),
then $P$ will be concentrated around its expectation,
assuming the variables $x_i$ are picked independently at random.

Because in our case the variables are not picked independently,
we need to use a different bound (other than that, the
 original Kim-Vu bound would be strong enough
for our purpose).
The bound we use requires the following concept 
of almost independence.
\begin{definition}
A distribution $\Pd_x$ over $\{0,1\}^\ell$ is
\emph{$(\delta,m)$-almost independent}
if for all sets $M$ of size $|M| < m$ and any $j \notin M$
\begin{align}
\Pr_{x\leftarrow \Pd_x}
[x_j = 1|\forall i \in M: x_i = 1] \leq 
\Pr_{x\leftarrow\Pd_x}[x_j = 1](1+\delta)
\end{align}
\end{definition}

We use the following bound, which is
proven in \cite{Holens11}.
It uses a technique first used implicitly by \cite{ScSiSr95}
and which was later used in \cite{Rao08} to prove concentration
bounds for parallel repetition, and by \cite{ImpKab10} to prove
constructive concentration results.

For a polynomial $P$ in variables $x_j$, and a distribution $\Pd_{x}$
over these variables, we let $\Pd_{x}^*$ be the distribution
obtained by picking each $x_j$ independently of the others, but
with the marginal distribution given by $\Pd_{x}$.
We then set
$\mu^* = \displaystyle\E_{x\leftarrow\Pd_{x}^*}[P(x)]$ and
$E^* = \displaystyle\max_{\emptyset \subsetneq B \subseteq \{x_1,\ldots,x_\ell\}}
\E_{x\leftarrow\Pd_{x}^*}[\del_{B} P(x)]$.

\begin{theorem}\label{thm:KimVuStyle}
Let $\Pd_{x}$ be
an $(\delta,rm)$-almost independent distribution
over $\{0,1\}^{\ell}$.
Let $P(x)$ be a polynomial of degree
at most $r$ in the variables $x_i$, i.e., 
$P(x) = \sum_{j=1}^{n} v_j$ with $v_j = \prod_{i\in e_j} v_j$, where 
$|e_j| \leq r$.

Then,
\begin{align}
\Pr_{x \leftarrow \Pd_{x}}\Bigl[P(x)\geq \mu^* (1+\epsilon)\Bigr] 
\leq 
\Bigl(\frac{(1+\delta)^r (1+\frac{r^rm^rE^*}{\mu^*})}{1+\epsilon}\Bigr)^m\;.
\end{align}
\end{theorem}

Using this bound, we can now prove Lemma~\ref{lem:concentration}.
\begin{proof}[Proof (of Lemma~\ref{lem:concentration})]
We use Theorem~\ref{thm:KimVuStyle} on the polynomial
$P_{w,y^*}$, where we set
$\delta = 1$, $\epsilon = 2^{\frac{9n}{100}}/\mu^*$ and $m = 2^{\frac{n}{100r}}$.
We note first that indeed the random variables $F_{(i,j)}$ are 
$(\delta,rm)$-independent: conditioning on $F_{(x,y)} = 1$
is the same as conditioning on $f(x) = y$, 
and so we can see that one needs to condition on at least
$2^{n-1}$ such events in order to double
the probability that $F_{(x,y)} = 1$ for any $(x,y)$.

Thus, Theorem~\ref{thm:KimVuStyle} yields:
\begin{align}
\Pr_{f\leftarrow \cP}[P_{w,y^*}(f)  \geq \mu^* + 2^{\frac{9n}{100}}]
&\leq
\Bigl(
\frac{2^r\max(2, \frac{2 r^r 2^{\frac{n}{100}} E^*}{\mu^*})}
{\frac{2^{\frac{9n}{100}}}{\mu^*}}
\Bigr)^{2^{n/100r}}\\
&\leq
\max
\Bigl(\frac{2\cdot2^{r}\mu^*}{2^{\frac{9n}{100}}}
,
\frac{2\cdot2^r r^r 2^{\frac{n}{100}} E^*}{2^{\frac{9n}{100}}}
\Bigr)^{2^{n/100r}}\\
&\leq
(\tfrac{1}{2})^{2^{n/100r}}\;,
\end{align}
where we applied Lemma~\ref{lem:derivativeLow} to bound both $\mu^*$
and $E^*$ in the last step.
An application of Lemma~\ref{lem:polynomialGivesABound} finishes the proof.
\end{proof}

\section{Non-uniform reductions and superpolynomial security}
\label{sec:nubb}

Theorem~\ref{thm:main} excludes the existence of a \emph{uniform}
black-box reduction constructing a pseudorandom generator
from a one-way function with few calls.
Potentially, one way to overcome this lower bound would be to 
give a non-uniform security reduction, in which
case the result would be weaker, but still very interesting.
Such non-uniform construction can be excluded by
the techniques given in \cite{GGKT05}, and we apply their 
technique here to prove that our lower bound applies to non-uniform
constructions as well.

Furthermore, we also generalize our results to one-way functions 
with different security.
\begin{definition}\label{def:nfbb}
A \emph{non-uniform fully black-box construction of a pseudorandom generator 
from a regular one-way function with security $s(k)$}
consists of two oracle algorithms $(g,A)$.
The \emph{construction} $g^{(f)}$ is a polynomial time 
oracle algorithm which provides, for each~$k$, 
a function $g_k: \{0,1\}^{m(k)}
\to \{0,1\}^{m'(k)}$ with $m'(k) > m(k)$.
For this, $g_k$ may call $f_k$ as an oracle, 
and $m(k),m'(k)$ may depend on $n(k)$ and $n'(k)$.

Further, the \emph{security reduction} 
$A^{(\cdot,\cdot)}(k,\cdot,\cdot)$ 
is an oracle algorithm which
does at most $s(k)$ queries,
and has the property that for any regular function $f$ and any 
oracle $B$ for which
\begin{align}\label{eq:2000}
	\Pr_{v,B}[B(k,g_k(v))=1] - \Pr_{w,B}[B(k,w)=1] \geq \frac{1}{100}
\end{align}
for infinitely many $k$, there is $h_k \in \{0,1\}^{s(k)}$ such that
\begin{align}\label{eq:2001}
\Pr_{x,A}[\text{$A^{(B,f)}(k,h_k,f_k(x))$ inverts $f_k$}]
> \frac{1}{s(k)}
\end{align}
for infinitely many $k$.

Similar to before, $A(k,\cdot,\cdot)$ only calls the 
oracles $f(k,\cdot)$ and 
$B(k,\cdot)$.
\end{definition}
In an actual reduction, one would of course excpect
that it works given a much weaker condition than (\ref{eq:2000}).
In particular, a reasonable reduction will invert $f$
with some probability if the constant $\frac{1}{100}$ is replaced
by any polynomial.
Excluding constructions which even adhere to 
Definition~\ref{def:nfbb} is of course then stronger.
\subsection{$\BreakOW$ does not non-uniformly invert}
We first show that no non-uniform oracle algorithm
with access to $\BreakOW$ inverts a random permutation $f$.
\begin{lemma}\label{lem:breakerHelpsNothingGT}
Let $g^{(\cdot)}: \{0,1\}^m \to \{0,1\}^{m*}$ be an 
$r$-query normalized oracle construction,
$\frac{n}{100r} \in \mathbb{N}$.
Fix an oracle function $C^{(\BreakOW,f)}(y)$ making at 
most $q < 2^{\frac{n}{10}}$ queries to its oracles.
Let $\cW$ be the set which contains all permutations $f \in \cP$ for 
which both
\begin{align}
\Pr_{y\leftarrow\{0,1\}^n}[C^{(\BreakOW,f)}(y) = f^{-1}(y)] 
&\geq 2^{-\frac{n}{20}}\qquad\qquad\text{, and}
\label{eq:17}\\
\forall w,y^*: |Q_{f,w,y^*}| &\leq 2^{\frac{n}{10}}\label{eq:18}
\end{align}
holds.  
Then, $\frac{|\cW|}{|\cP|}  \leq 
2^{- 2^{n/2}}$.
\end{lemma}
\begin{proof}
As in \cite{GGKT05}, we find an encoding of $f$ which is
$2^{\frac{n}{2}}$ bits shorter than $\log(2^n!)$ (the minimal
length of a bitstring needed to describe an arbitrary
permutation on $2^n$ elements).
The encoding has the property that $f$ can be recovered (given~$C$)
from it.
Actually, it is somewhat
easier to describe the encoding simply as a injective
function $\cF$ mapping
onto a set with fewer than $(2^n!)2^{-2^{n/2}}$ elements,
which is of course equivalent.

Fix some function $f$ which satisfies both (\ref{eq:17}) and (\ref{eq:18}).
We first find a large subset~$S$ of the images of $f$, 
which has the property that $\cF$ does not need to 
describe how $f$ maps elements of $f^{-1}(S)$ to $S$,
and yet $\cF$ will be injective.
For this, we first modify $C$ such that whenever it queries 
$v=\BreakOW^{(f)}(w)$, 
it afterwards evaluates $g^{(f)}(v)$ on the result (unless 
$\BreakOW$ returned $\bot$).
Then, the following algorithm outputs~$S$.

\mbox{}\begin{mdframed}\setlength{\parindent}{0cm}%
\noindent \textbf{Algorithm} BuildSets$(f)$
\medskip

\hrule

\medskip

Modify $C$ as in the text\\
$I := \{y : C^{(\BreakOW,f)}(y) = f^{-1}(y)\}$\\
\textbf{while} $I \neq \emptyset$ \textbf{do}\\
\mbox\qquad $y^* \leftarrow I$ \centerComment{An arbitrary element of $I$}\\
\mbox\qquad $S := S \cup \{y^* \}$\\
\mbox\qquad Let $Q$ be the answers of $f$ to the
 queries done by $C^{(\BreakOW,f)}(y^*)$.\\
\mbox\qquad \textbf{for} $x^* \in \{0,1\}^n$ \textbf{do}\\
\mbox\qquad\qquad $f^* := f_{(x^*,y^*)}$\\
\mbox\qquad\qquad \textbf{if} there is $w$ such that $C^{(\BreakOW,f)}(y^*)$ 
calls $\BreakOW^{(f)}(w)$ and\\
\mbox\qquad\qquad\qquad\qquad
$\BreakOW^{(f)}(w) \neq \BreakOW^{(f^*)}(w)$ \textbf{then}\\
\mbox\qquad\qquad\qquad $Q := Q \cup \QueryY(g,\BreakOW^{(f^*)}(w),f)$\\
\mbox\qquad $I := I \setminus (Q \cup \{y^*\})$\\
\textbf{return} $S$
\end{mdframed}

\medskip

We show that $|S| \geq s := \frac{2^{19n/20}}{2^{n/10}2^{n/5}+r2^{n/10}+1}
\geq 2^{\frac{3n}{5}}$.
First, from (\ref{eq:17}) we see
that $|I| \geq 2^{\frac{19n}{20}}$.
We claim that for each $y$, $Q$ has size at most
$|Q| \leq 2^{\frac{n}{10}} 2^{\frac{n}{5}} + r 2^{\frac{n}{10}}$ 
when it is removed from $I$.
We get this since $C$ makes at most~$2^{\frac{n}{10}}$ calls 
to $\BreakOW(w)$, 
and Lemma~\ref{lem:breakerStaysConstant} implies that
for each of these calls, there can be at most 
 $2^{\frac{n}{5}}$ elements $x^*$ for which $\BreakOW^{(f^*)}(w)
\neq \BreakOW^{(f)}(w)$.
Further, $g$ makes at most $r 2^{\frac{n}{10}}$ calls to~$f$ (due 
to our modification above this is a bit larger than $2^{\frac{n}{10}}$).

Let now $S' \subseteq S$ be some subset of size 
$s$, set
$t = 2^{n} - s$, and let 
 $x_0,\ldots,x_{t-1}$ be the elements of $\{0,1\}^n$ which
are \emph{not} preimages of elements in $S'$, in lexicographic order.
We show in the next paragraph that the map $\cF$ which maps 
$f \mapsto (x_0,f(x_0),\ldots,x_{t-1},f(x_{t-1}))$
is injective.
The number of possible images can be counted by
first considering the possible sets $\{x_0,\ldots,x_{t-1}\}$ and
$\{f(x_0),\ldots,f(x_{t-1})\}$ (there are $\binom{2^n}{t} = 
\binom{2^{n}}{s}$ of those) and then considering the $t!$ permutations
from the first to the second set, which shows that 
\begin{align}
\frac{|\cW|}{|\cP|}
&\leq
\binom{2^n}{s}^2 \frac{(2^n-s)!}{2^n!} 
=
\binom{2^n}{s}
\frac{1}{s!}
\leq
\Bigl(\frac{e^2 2^n }{s^2}\Bigr)^{s}
\leq
2^{-s}\;.
\end{align}

It remains to show that the map is injective.
To see this, suppose that $f_1 \neq f_2$ satisfy $\cF(f_1) = \cF(f_2)$.
Then, $f_1^{-1}(y) = f_2^{-1}(y)$ for all $y$ for which
$f_1^{-1}(y) \neq C^{(\BreakOW,f_1)}(y)$ (the pair $(f_1^{-1}(y),y)$
appears in $\cF(f_1)$, and so it must also appear in $\cF(f_2)$).
Since this holds analogous for $f_2$, 
there must be $(x_1,x_2,y)$ such that $x_1 = C^{(\BreakOW,f_1)}(y) = f_1^{-1}(y) 
\neq f_2^{-1}(y) = C^{(\BreakOW,f_2)}(y) = x_2$.

Since the two answers of $C(y)$ differ in the two runs, there must be
some call of~$C$ to an oracle with the same input, 
but for which the two answers differ.
This cannot be a call to the oracle~$f_1$ or $f_2$, 
as otherwise this call would
appear in $\cF(f_1)$ and in $\cF(f_2)$.
Thus, it must be that some for some $w$ we have $v_1 := \BreakOW^{(f_1)}(w) \neq
\BreakOW^{(f_2)}(w) =: v_2$, and $\BreakOW(w)$ is actually
called by $C$ in both experiments.
Suppose without loss of generality that $v_1$ occurs first
in the enumeration within $\BreakOW$.
Then, one of the queries which $g^{(f_1)}(v_1)$ does must have answer $y$,
as otherwise all elements of $\Query(g,v_1,f_1)$ would appear in $\cF(f_1)$,
and so $\BreakOW(w)^{(f_2)}(w) = v_1$ as well.
Thus, one of the answers was~$y$, and since the other answers appear 
in $\cF(f_1)$, $f_2$ and $f_1$ behave the same for these answers.
But this implies that $\BreakOW^{(f_2^*)}(w) \neq \BreakOW^{(f_2)}(w)$,
where $f_2^* = (f_2)_{x_1,y}$, and so $(x_1,f_2(x_1))$ must appear 
in $\cF(f_2)$, which contradicts $\cF(f_1) = \cF(f_2)$.
\end{proof}

\subsection{Non-uniform black-box separation}
We can now prove Theorem~\ref{thm:mainB}, which we restate
for convenience.

\maintheoremB*
\begin{proof}
As previously, we use Theorem~\ref{thm:prgandPUOWF} and 
thus assume we have a non-uniform fully black-box reduction
which yields a pseudouniform one-way function.

Thus, we suppose we are given $(g,A)$.
Again we set $n(k) := n'(k) := k$, and let $m(k)$ be the input
length of $g$ as provided by the reduction.
Let $r(k)$
the number of calls to $f$.
As before we assume that $\frac{n(k)}{100r(k)} \in \bbN$,
and  modify $g$ so it is normalized.

For all $k$ with $r(k) \geq 
\frac{n(k)}{1000 \log(s(k))}$ we let
$\Breaker_k$ be the function which always outputs $0$
and $f_k$ a permutation which is one-way
against circuits of size $2^{\frac{n}{5}}$, which exists by \cite{GGKT05}.

Otherwise, we consider $p(g_k)$. 
If $p(g_k) \geq \frac12$ we 
set $\Breaker_k$ to be $\BreakOW$.
Lemma~\ref{lem:breakOWHelpsInvertG} again implies that 
$\BreakOW$ helps to invert $g$, but now we apply 
Lemma~\ref{lem:breakerHelpsNothingGT} and the union bound
to get a function which is hard to invert for all $h_k$.

If $p(g_k) \leq \frac12$, Lemma~\ref{lem:largeP}
gives a set  $W_k \subseteq \{0,1\}^{m(k)}$ and 
$\cY_k \subseteq \{0,1\}^{n(k)}$ for which the output
of $g^{(f)}$ is likely distinguished from uniform by $\BreakPU$.

Writing the function $f$ as $f = \pi \circ p$ for some 
a random permutation $\pi$ on $\cY$ and a regular function $p: \{0,1\}^n
\to \cY$ we can apply Theorem~1 of \cite{GGKT05} (which also holds 
if the circuit has oracle gates to $\BreakPU$).
We can thus find a function which is hard for $A$ and any 
advice string $h_k$, and yet
the output $g^{(f)}$ will be distinguished from uniform.
\end{proof}

\section{Non-security parameter restricted constructions}
\label{sec:nsprc}





\subsection{Fixing the polynomial in the construction}

Suppose that we have given a black-box construction $(g,A)$ of a 
pseudouniform one-way function from a one-way function together with its 
security reduction.
The requirement on the efficiency of the construction
is that for every choice of $(f,\Breaker)$, 
both $g$ and $A$ should run in polynomial time.
In other words, for any $(f,\Breaker)$ there should be $c \in \bbN$
such that $f(k,\cdot)$ and $\Breaker(k,\cdot)$ run in time $k^c$.
Note that~$c$ can depend on $f$ and $\Breaker$.

There do exist constructions $(g,A)$ which are polynomial
for any $(f,\Breaker)$, but where $c$ indeed 
depends inherently on the oracle.\footnote{An example follows:
suppose the function $g(k,v)$ first queries 
$f(0,\mathbf{0}), f(1,\mathbf{0}),\ldots,f(\log(\ell),\mathbf{0})$.
If all answers were the $0$-string, execute some algorithm
which runs in linear time.
Otherwise, let $c'$ be the index of the first answer which differs,
and execute an algorithm which runs in time $k^{c'+1}$.}
However, it turns out that it is always possible to fix \emph{finitely} many 
outputs of the oracles
$f$ and $\Breaker$ such that after fixing these, $c$ is independent
of the choice of the remaining positions.

A \emph{prefix} $(f^*,\Breaker^*)$ is simply the truth table for these
oracles for lengths up to some integer $k_0$; oracles $(f,\Breaker)$
\emph{agree} with the prefix if their truth table up to length $k_0$ equals
the one given by the prefix.
A prefix $(f^{(2)},\Breaker^{(2)})$ \emph{extends}
a prefix $(f^{(1)},\Breaker^{(1)})$ if 
the truth table of $(f^{(2)},\Breaker^{(2)})$ is larger than the truth
table of $(f^{(1)},\Breaker^{(1)})$, and they agree everywhere where 
$(f^{(1)},\Breaker^{(1)})$ is defined.

\begin{lemma}\label{lem:fixC}
Suppose a black-box reduction $(g^{(f)},A^{(\Breaker,f)})$ is given,
and fix some length function $n(k)$.
There exists a prefix $(f^*,\Breaker^*)$ and $c \in \bbN$
such that 
for any pair $(f,\Breaker)$ which agrees with $(f^*,\Breaker^*)$ 
 we have the following properties:
\begin{enumerate}
\item $g^{(f)}(\ell,\cdot)$ makes at most $\ell^c$ queries to
$f(k,\cdot)$, and all of these queries satisfy $k \leq \ell^c$
\item $A^{(\Breaker,f)}(k,w)$ makes
at most $k^c$ queries to $\Breaker(\ell,\cdot)$, and all of these
queries satisfy $\ell \leq k^c$
\item $A^{(\Breaker,f)}(k,w)$ makes at most $k^c$ queries
to $f(k',\cdot)$, and all of these
queries satisfy $k' \leq k^c$.
\end{enumerate}
\end{lemma}
\begin{proof}
Suppose not, let $d \in \bbN$, and
suppose we have given any prefix $(f^{(d)},\Breaker^{(d)})$.
Then, there exists a pair $(f,\Breaker)$ of oracles which agree 
with $(f^{(d)},\Breaker^{(d)})$ and where one of 1, 2, or 3 is 
violated for $c = (d+1)$.
Fix a length $k$ for which this is violated, and find a prefix 
$(f^{(d+1)},\Breaker^{(d+1)})$ of $(f,\Breaker)$ such
that all queries done for up to security parameter $k$ 
are fixed in the prefix.

Thus, there is an infinite sequence of 
prefixes $\{(f^{(i)},\Breaker^{(i)})\}_{i\geq0}$ such that
$(f^{(i+1)},\Breaker^{(i+1)})$ extends $(f^{(i)},\Breaker^{(i)})$, and 
for any $d \in \bbN$ there is an input which violates one of the 
conclusions of the lemma.

Clearly, such an infinite sequence defines a pair $(f,\Breaker)$ 
for which $(g,A)$ is not polynomial.
\end{proof}

\subsection{Excluding general reductions}
We now come to the proof of Theorem~\ref{thm:mainC},
which we restate for convenience.

\maintheoremC*

\paragraph{Preparations for the proof}
As before, we prove the analogous statement 
for pseudouniform one-way functions.
Also, we assume that the theorem is not true, and 
that we have given $(g,A)$, and show that 
we can find oracles $(f,\Breaker)$ which contradict
the assumption that $(g,A)$ is a fully black-box construction.

We can assume that $g(\ell,v)$ never queries $f(k,x)$ twice for any $(k,x)$.
Also, we use Lemma~\ref{lem:fixC}, 
which fixes a prefix for $(f,\Breaker)$ and gives us a constant $c$
for which the properties in Lemma~\ref{lem:fixC} are satisfied, 
we will use this constant throughout the proof.

We now choose $k_0$ and set $\ell_0 = k_0^c$ such that neither
$f(k_0,\cdot)$ nor $\Breaker(\ell_0,\cdot)$ has been
defined by Lemma~\ref{lem:fixC}.
Furthermore, define all oracles $\Breaker(\ell,\cdot)$ for $\ell < \ell_0$
which have not been defined yet 
to oracles which do nothing (i.e., constantly output $\bot$).
Analogously, define all oracles $f(k,\cdot)$ for $k < k_0$ which have not been
defined yet to random permutations of length $n(k)$.

Next, we pick a constant $\ctilde$ for later.
We require that it satisfies that for any  
$k \in \{\ell^{1/c},\ldots,\ell^{c}\}$ we have
$\ell^{1/\ctilde} \leq n(k) \leq \ell^{\ctilde}$; this is possible because
$n(k)$ is a length function.

\paragraph{Overview and some basics of the proof}
In the main part of the proof, we
define the oracles $\Breaker(\ell,\cdot)$ and $f(k,\cdot)$.
We essentially use one iteration for each $\ell$, and increase
$\ell$ over time.
At the beginning of iteration $\ell$, we will have defined the oracles
$\Breaker(1,\cdot),\ldots,\Breaker(\ell-1,\cdot)$
and $f(k,\cdot)$ for any $k < \ell^{1/c}$.

At this point, we enumerate each $n$ 
which is possibly the length of the shortest query made
by $g(\ell,\cdot)$, ignoring the length of those for which $f(k,\cdot)$
has been defined already.

For each such $n$, we consider the probability 
\begin{align*}
q_{\ell,n} &:= 
\Pr_{v,f}\Bigl[\text{$g^{(f)}(\ell,v)$ 
queries $f$ on security parameters $k \geq \ell^{1/c}$
a total  of at most }\\
&\qquad\qquad \text{%
$\tfrac{n}{d\log(n)}$ times,
and for all these queries $f(k,\cdot)$ we have $n(k) \geq n$}\Bigr]\,
\end{align*}
where $f(k,\cdot)$ is chosen as random permutation for any $k \geq \ell^{1/c}$.
The parameter $d$ will be defined later, and is 
slowly growing as $\ell \to \infty$.

We then distinguish two cases:
The first case is if $q_{\ell,n} \leq \ell^{-\ctilde-2}$ for all $n$.

In this case, we define $\Breaker(\ell,\cdot) := \bot$, so that
it does nothing on this length, increase $\ell$, and go to the 
next iteration.
We will show that infinitely often $q_{\ell,n}$ must be 
larger than $\ell^{-\ctilde-2}$ for some $n$, 
as otherwise we can obtain an
oracle $(f,\Breaker)$ for which $r_f \in \Omega(n_f^{-}/\log(n_f^{-}))$.

The second case is more interesting: there is
$\ntilde$ for which $q_{\ell,\ntilde} > \ell^{-\ctilde-2}$. 

In this case, we know that with some polynomial probability, 
$g(\ell,\cdot)$ will only make few queries.
We would like to apply the previous machinery,
but cannot do so directly: $g$ possibly makes more 
than $\ntilde/d\log(\ntilde)$
many queries for some oracle $f$, and possibly queries $f$ on input lengths 
shorter than $\ntilde$ for some oracle $f$.

Also (and this is the problem we fix first),
the previous machinery only allows $g$ to make queries to one fixed input 
length $\ntilde$, whereas $g$ may query $f$ with many different parameters
$k$ for which $n(k) \geq \ntilde$.

To solve this, we use the following idea: underlying to 
$f(k,\cdot)$ could in fact be 
a \emph{single} one-way function 
$\ftilde: \{0,1\}^{\ntilde} \to \{0,1\}^{\ntilde}$
for many different values of $k$, so that $f(k,x) = S_k(\ftilde(P_k(x)))$
for some simple to compute 
projection $P_k: \{0,1\}^{n(k)} \to \{0,1\}^{\ntilde}$
and some expansion $S_k : \{0,1\}^{\ntilde} \to \{0,1\}^{n(k)}$.

Thus, we pick uniform random injective functions $P_k$ 
and uniform injective expansions $S_k$ for each $k$
for which $n(k) \geq \ntilde$.
We then consider the construction
$\gtilde_{n}^{(\ftilde,f,P,S)}$.
This construction is is defined as follows:
\begin{quote}
The function $\gtilde_{n}^{(\ftilde,f,P,S)}$ simulates $g$,
except whenever $g$ calls the oracle $f$.

In case $g$ calls $f(k,x)$ for some $k$ with $k < \ell^{1/c}$, 
the answer of $f(k,x)$ is hard-coded into $\gtilde$ (because
$f$ is already defined on these lengths).

If $g$ calls $f(k,x)$ for some $k$ with $n(k) < \ntilde$ and $k \geq \ell^{1/c}$,
then $\gtilde$ calls $f(k,x)$ as well.

In case $g$ calls $f(k,x)$ for some $k$ with $n(k) \geq \ntilde$, 
and $k \geq \ell^{1/c}$,
$\gtilde_{n}$ instead calls $S_k(\ftilde(P_k(x)))$.
\end{quote}
The function $\gtilde$ behaves almost as $g$ when
$\ftilde$ is chosen as a random permutation.
The only exception is 
in the unlikely case that $P_k(x) = P_k(x')$ for two 
queries $(k,x) \neq (k,x')$ to $f$.

The function $\gtilde$ solves the last problem above, so that we get
closer to apply the previous machinery.
However, $\gtilde$ still can make more than $\ntilde/d \log(\ntilde)$ queries
to $f$ or query $f$ on shorter inputs than $\ntilde$ for some $f$.
Thus, we consider the construction $\htilde$, which is just like
$\gtilde$, but has additional restrictions:
\begin{quote}
Whenever $\gtilde$ does more than $\ntilde/d\log(\ntilde)$ calls, 
$\htilde$ simply stops and outputs $(v,\bot)$.

Whenever $g$ does a call to $f(k,\cdot)$ with $n(k) < \ntilde$
and $k \geq k_0$, $\htilde$ stops and outputs $(v,\bot)$.
\end{quote}
We note that $\htilde$ does not need to call the oracle $f$ at any time 
anymore.

As long as $d \to \infty$ for $\ell \to \infty$, the results
from the previous sections will guarantee that breaking the
pseudouniformity of $\htilde$ does not help inverting $\ftilde$.
The main difficulty is that $\htilde$ may behave very differently from
$\gtilde$.
However, we can note that
\begin{align}\label{eq:24}
\Pr_{\ftilde,P,S,v,f}[\htilde^{\ftilde,P,S}(v) = \gtilde^{\ftilde,f,P,S}(v)]
\geq q_{\ell,\ntilde} - \frac{\ntilde^2}{2^{\ntilde}}\,
\end{align}
because as long as no two queries to $P_k(\cdot)$ collide
in the evaluation of $\gtilde$, each query will be
answered with a uniform random answer, and so 
$g$ and $\gtilde$ will behave exactly the same.

We are now interested in the probability that 
$\BreakOW$ inverts a random image of $\gtilde$.
To apply the previous machinery, we want to instantiate \emph{$\BreakOW$ 
using $\htilde$}.
Thus, we consider the probability
\begin{align}
p_{\ell,\ntilde} &:= 
\Pr_{\ftilde,P,S,v,f} 
[\SafeToAnswer_{\htilde}(\htilde^{\ftilde,P,S}(v),
\QueryY(\htilde,\ftilde,v)) {} \land {} \\
&\qquad\qquad\qquad\qquad
\htilde^{\ftilde,P,S}(v) = 
\gtilde^{\ftilde,f,P,S}(v)]\;.\nonumber
\end{align}
Here, $\SafeToAnswer$ is instantiated using $\htilde$ instead 
of $g$.\footnote{Strictly speaking, to instantiate
$\SafeToAnswer$ we should give it a function $\htilde$ which only
uses the oracle $\ftilde$, but not oracles $P$ and $S$.
For that purpose, one can think of $P$ and $S$ as being hardcoded into
$\htilde$.}

We will then show that we can do a similar case distinction 
$p_{\ell,\ntilde}$ as we did in the previous 
sections on $p(g)$.
This will allow us to build oracles $(f(k,\cdot),\Breaker(\ell,\cdot))$ 
where $\Breaker(\ell,\cdot)$ breaks the construction on this length.

After this, we set $\Breaker(\ell',\cdot) := \bot$ for 
$\ell < \ell' < \ell^{c^2}$, which ensures that there is no 
problem because different lengths are interfering with each other.
We then go to the next iteration for which $\Breaker(\ell,\cdot)$
is not yet defined.

\paragraph{Building the oracles}

We now describe a randomized procedure which builds 
oracles $f$ and $\Breaker$ by building a sequence of extending prefixes 
(as in the proof of Lemma~\ref{lem:fixC}).
After this, we prove that the oracle arising from this sequence has the
required properties with probability $1$.

\mbox{}\begin{mdframed}\setlength{\parindent}{0cm}%
\noindent \textbf{Algorithm} GenerateOracles
\medskip

\hrule
\medskip

Fix $\Breaker$ and $f$ up to some length using Lemma~\ref{lem:fixC}, 
then ensure that\\
\mbox\qquad 
$f(k,\cdot)$ is defined up to security parameter $k_0$ for some $k_0$, 
and that\\
\mbox\qquad $\Breaker(\ell,\cdot)$ is defined up to $k_0^c$.\\
$d := 1$\\
$\ell := $ smallest $\ell$ for which $\Breaker(\ell,\cdot)$ has not yet
been defined\\
\textbf{do forever}\\ 
\mbox\qquad /\!\!/ \emph{We define $\Breaker(\ell,\cdot)$ in this iteration}\\
\mbox\qquad \textbf{if} $\forall n \in \{\ell^{1/\ctilde},\ldots,
\ell^{\ctilde}\}$: 
$q_{\ell,n} \leq \ell^{-(\ctilde+2)}$ \textbf{then} \\
\mbox\qquad\qquad $\Breaker(\ell,\cdot) := \bot$ \centerComment{$\Breaker$
will not help on this length}\\
\mbox\qquad\qquad \textbf{if} $\exists k \in \bbN: k^c = \ell$ \textbf{then}\\
\mbox\qquad\qquad\qquad $f(k,\cdot) \leftarrow \Pi_{n(k)}$ 
\centerComment{A random permutation on $n(k)$ bits}\\
\mbox\qquad\textbf{else} \\
\mbox\qquad\qquad let $\ntilde$ be such that $q_{\ell,\ntilde} >\ell^{-(\ctilde+2)}$\\
\mbox\qquad\qquad \textbf{if} $p_{\ell,\ntilde} \leq \frac{1}{2}\ell^{-(\ctilde+2)}$ \textbf{then}\\
\mbox\qquad\qquad\qquad \textbf{try at most} $\ell^{\ctilde+3}$ \textbf{times}\\
\mbox\qquad\qquad\qquad\qquad BuildPUBreaker$(\ell,\ntilde)$\\
\mbox\qquad\qquad\qquad \textbf{stop if} $\Breaker(\ell,\cdot)$ has distinguishing advantage
at least $\ell^{-(\ctilde+3)},$\\
\mbox\qquad\qquad\qquad \phantom{\textbf{stop if}} otherwise roll back the changes and try the loop again\\
\mbox\qquad\qquad \textbf{else}\\
\mbox\qquad\qquad\qquad \textbf{try at most} $\ell^{\ctilde+3}$ \textbf{times}\\
\mbox\qquad\qquad\qquad\qquad BuildOWBreaker$(\ell,\ntilde)$\\
\mbox\qquad\qquad\qquad \textbf{stop if} $\Breaker(\ell,\cdot)$ inverts $g$ with
probability at least $\ell^{-(\ctilde+3)}$,\\
\mbox\qquad\qquad\qquad \phantom{\textbf{stop if}} otherwise roll back the changes and try the loop again\\
\mbox\qquad\qquad\textbf{fi}\\
\mbox\qquad\qquad $d := d+1$\\
\mbox\qquad $\ell := \ell+1$
\end{mdframed}

\medskip

\mbox{}\begin{mdframed}\setlength{\parindent}{0cm}%
\noindent
\textbf{procedure} BuildPUBreaker($\ell,\ntilde$)\\
\mbox\qquad $r := \ntilde/d\log(\ntilde)$\\
\mbox\qquad Pick $\cY \subseteq \{0,1\}^{\ntilde}$, $|\cY| = 2^{\ntilde/100r}$ 
u.a.r.\\
\mbox\qquad \textbf{for each} $k \in \{\ell^{1/c},\ldots,\ell^{c}\}$ 
with $n(k) \geq \ntilde$ \textbf{do}\\
\mbox\qquad\qquad pick a regular function 
$P_k : \{0,1\}^{n(k)} \to \{0,1\}^{\ntilde}$ u.a.r.\\
\mbox\qquad\qquad pick an injective function $S_k : \{0,1\}^{\ntilde} \to 
\{0,1\}^{n(k)}$ u.a.r.\\
\mbox\qquad /\!\!/ \emph{At this point, $\htilde$ and $\cY$ are defined}\\
\mbox\qquad Obtain $W$ (using $\htilde$ as underlying function) as in Lemma~\ref{lem:largeP}.\\
\mbox\qquad $\Breaker(\ell,\cdot) := \BreakPU(W)$\\
\mbox\qquad Pick a regular function $\ftilde: \{0,1\}^{\ntilde} \to \cY$ u.a.r.\\
\mbox\qquad \textbf{for each} $k \in \{\ell^{1/c},\ldots,\ell^{c}\}$ \textbf{do}\\
\mbox\qquad\qquad \textbf{if} $n(k) \geq \ntilde$ \textbf{then}\\
\mbox\qquad\qquad\qquad $f(k,\cdot) := S_k \circ \ftilde \circ P_k$\\
\mbox\qquad\qquad \textbf{else}\\
\mbox\qquad\qquad\qquad $f(k,\cdot) \leftarrow \Pi_{n(k)}$\\
\mbox\qquad \textbf{for} $\ell' \in \{\ell+1,\ldots,\ell^{c^2}\}$ \textbf{do}\\
\mbox\qquad\qquad $\Breaker(\ell',\cdot) := \bot$\\
\mbox\qquad $\ell := \ell^{c^2}+1$
\end{mdframed}

\medskip

\mbox{}\begin{mdframed}\setlength{\parindent}{0cm}%
\noindent
\textbf{procedure} BuildOWBreaker($\ell,\ntilde$)\\
\mbox\qquad $r := \ntilde/d\log(\ntilde)$\\
\mbox\qquad \textbf{for each} $k \in \{\ell^{1/c},\ldots,\ell^{c}\}$ 
with $n(k) \geq \ntilde$ \textbf{do}\\
\mbox\qquad\qquad pick a regular function 
$P_k : \{0,1\}^{n(k)} \to \{0,1\}^{\ntilde}$ u.a.r.\\
\mbox\qquad\qquad pick an injective function $S_k : \{0,1\}^{\ntilde} \to 
\{0,1\}^{n(k)}$ u.a.r.\\
\mbox\qquad /\!\!/ \emph{At this point, $\htilde$ is defined}\\
\mbox\qquad $\Breaker(\ell,\cdot) := \BreakOW_{\htilde}^{(\ftilde)}(\cdot)$\\
\mbox\qquad Pick a permutation $\ftilde: \{0,1\}^{\ntilde} 
\to \{0,1\}^{\ntilde}$ u.a.r.\\
\mbox\qquad \textbf{for each} $k \in \{\ell^{1/c},\ldots,\ell^{c}\}$ \textbf{do}\\
\mbox\qquad\qquad \textbf{if} $n(k) \geq \ntilde$ \textbf{then}\\
\mbox\qquad\qquad\qquad $f(k,\cdot) := S_k \circ \ftilde \circ P_k$\\
\mbox\qquad\qquad \textbf{else}\\
\mbox\qquad\qquad\qquad $f(k,\cdot) \leftarrow \Pi_{n(k)}$\\
\mbox\qquad \textbf{for} $\ell' \in \{\ell+1,\ldots,\ell^{c^2}\}$ \textbf{do}\\
\mbox\qquad\qquad $\Breaker(\ell',\cdot) := \bot$\\
\mbox\qquad $\ell := \ell^{c^2+1}$
\end{mdframed}

\medskip

Clearly, GenerateOracles defines an infinite sequence of 
prefixes $(f^{(i)},\Breaker^{(i)})$, and as before we can extend 
that to a single oracle $(f,\Breaker)$.
Analogously, we can extend events which are defined on prefixes to an
infinite sequence of events.

\bigskip

We next explain how these procedures make their random choices.
For this we assume that for each $k \in \bbN$, 
a permutation $\overline{f}(k,\cdot)$ is picked.
Whenever GenerateOracles executes the assignment $f(k,\cdot) \leftarrow
\Pi_{n(k)}$ (in the part where it defines $\Breaker(\ell,\cdot) := \bot$,
it assigns $f(k,x) := \overline{f}(k,x)$).
We can imagine these permutations to be picked before GenerateOracles
is executed (in that way, we can talk about future assignments).

Also, for each $\ell$ and each $k \in \{\ell^{1/c},\ldots,\ell^{c}\}$,
we pick $\ell^{c+3}$ choices for $S_k$, $P_k$.
Also, for each possible $\ntilde$
we pick $\ell^{c+3}$ choices for $\ftilde$.
Then, in the $i$th iteration, we simply assume that the $i$th such 
choice is used.
As with $\overline{f}$, this is useful in order
to argue about future assignments.

\begin{lemma}\label{lemma:BreakerIsNotBot}
Consider an execution of the algorithm GenerateOracles.
For each $\ell$, let $N_\ell$ be the event that
the else clause of algorithm GenerateOracles is executed
on iteration $\ell$.

Then, with probability $1$, infinitely many events $N_\ell$ occur.
\end{lemma}
\begin{proof}
We let $d_\ell$ be the random variable which takes
the value of $d$ in the $\ell$th iteration of GenerateOracles.

For each $\ell$ and each $n \in \{1,\ldots,\ell^c\}$
we now define an event $B_{\ell,n}$.
For this event, we stop the normal execution of GenerateOracles
at loop $\ell$, and instead extend $f$ using $\overline{f}$ 
exclusively (i.e., fill everything with the random permutations
we picked before).
We then 
let $B_{\ell,n}$ be the event which occurs if $r_f \leq 
n/(d_\ell \log(n))$ and $n_{f}^- = n$ in this extension.
Clearly $\Pr[B_{\ell,n} \land \lnot N_\ell] \leq 
\Pr[B_{\ell,n} | \lnot N_\ell]\leq \ell^{-(\ctilde+2)}$,
because in case we extend with random permutations, $q_{\ell,n}$ 
is defined exactly as the probability that the event $B_{\ell,n}$ occurs.

Thus, $\sum_{\ell, n \leq \ell^{\ctilde}} \Pr[B_{\ell,n} \land \lnot N_{\ell}] < \infty$,
and so by the Borel-Cantelli lemma,
$(B_{\ell,n} \land \lnot N_\ell)$ happens for infinitely many $\ell$ 
with probability $0$.

Now, suppose that in some execution only finitely many events $B_{\ell,n}$
happen.
Then we found an oracle for which $r_f \in \Omega(n_f^{-}/\log(n^{-}_f))$,
because in this case we \emph{do} extended only using $\overline{f}$ 
starting from some fixed length.

Therefore, in all executions infinitely many events $B_{\ell,n}$
happen, and so the event $N_\ell$ must happen for infinitely many $\ell$
with probability~$1$.
\end{proof}

\begin{lemma}\label{lemma:PUDistinguishes}
Suppose that $q_{\ell,\ntilde} > \ell^{-(\ctilde+2)}$,
$p_{\ell,\ntilde} \leq \frac{1}{2}\ell^{-(\ctilde+2)}$, $d > 2c$, and $\ell$ is
larger than some constant in an execution of the loop
in algorithm GenerateOracles.

Then, with probability at least $\frac{1}{8}\ell^{-(\ctilde+2)}$, after a single call 
to BuildPUBreaker$(\ell,\ntilde)$,
the oracle $\BreakPU(\ell,\cdot)$ has
advantage at least $\frac{1}{8}\ell^{-(\ctilde+2)}$ 
in distinguishing $g(\ell,v)$
from a uniform random string.
\end{lemma}
\begin{proof}
We first notice that $\Pr_{w\leftarrow \{0,1\}^{m(\ell)}}[\BreakPU(\ell, w) = 1]$
is negligible:
Lemma~\ref{lem:largeP} gives a set 
$W$ of size at most $2^{m(\ell) - \frac{\ntilde}{100}}$
and $\ntilde \geq \ell^{-\ctilde}$.

We next show that
\begin{align}\label{eq:22}
\Pr_{v,\ftilde',f,P,S}[\gtilde^{\ftilde',f,P,S}(v) \in W] 
\geq \frac{1}{3}\ell^{(\ctilde+2)}\;,
\end{align}
where $\ftilde'$ is chosen as a random function 
$\ftilde' : \{0,1\}^{\ntilde} \to \cY$, for $\cY$ of size $2^{\ntilde/100r}$
chosen uniformly at random.
From (\ref{eq:22}) we get the lemma by applying Markov's inequality.

To see (\ref{eq:22}), we use that
\begin{align*}
\Pr_{v,\ftilde',f,P,S}[\gtilde^{\ftilde',f,P,S}(v) \in W] 
&\geq
\Pr_{v,\ftilde,f,P,S}[\gtilde^{\ftilde,f,P,S}(v) \in W] - \frac{r^2}{|\cY|}\;,
\end{align*}
where $\ftilde \in \cP_{\ntilde}$ is a uniform random 
permutation on $\ntilde$ bits: this follows as in the proof of
Lemma~\ref{lem:largeP}.

We now see that
\begin{align*}
\Pr_{v,\ftilde,f,P,S}[&\gtilde^{\ftilde,f,P,S}(v) \in W]\\
&\geq
\Pr_{v,\ftilde,f,P,S}[\htilde^{\ftilde,P,S}(v) \in W \land 
 \htilde^{\ftilde,P,S}(v) = \gtilde^{\ftilde,f,P,S}(v)] \\
&\geq 
 \Pr_{v,\ftilde,f,P,S}\Bigl[\bigl(\lnot 
 \SafeToAnswer_{\htilde}(\htilde^{\ftilde,P,S}(v), \QueryY(\htilde,\ftilde,v))\bigr) \land
\htilde^{\ftilde,P,S}(v) = \gtilde^{\ftilde,f,P,S}(v)\Bigr],
 \end{align*}
due to the definition of $W$ in the proof of Lemma~\ref{lem:largeP}.

Using~(\ref{eq:24}), we see that this last probability
is at least $q_{\ell,\ntilde} - p_{\ell,\ntilde} - \frac{\ntilde^2}{2^{\ntilde}}$,
which gives (\ref{eq:22}), and therefore the lemma.
\end{proof}

\begin{lemma}\label{lemma:OWBInverts}
Suppose that $p_{\ell,\ntilde} \geq \frac{1}{2}\ell^{-(\ctilde+2)}$.
Then, with probability at least $\frac{1}{4}\ell^{-(\ctilde+2)}$,
after a call to BuildOWBreaker$(\ell,\ntilde)$, the oracle
$\BreakOW$ will invert $g(\ell,v)$ with probability at 
least $\frac{1}{4}\ell^{-(\ctilde+2)}$.
\end{lemma}
\begin{proof}
Consider, for fixed $(\ftilde,f,P,S)$ the probability
that 
\begin{align}
p' := \Pr_{v}[\SafeToAnswer_{\htilde}(\htilde^{\ftilde,P,S}(v),
\QueryY(\htilde,\ftilde,v)) \land \htilde^{\ftilde,P,S}(v) = 
\gtilde^{\ftilde,f,P,S}(v)]\;.
\end{align}
We know that $p_{\ell,\ntilde} = \E_{\ftilde,P,S,f}[p'] \geq
\frac{1}{2}\ell^{-(\ctilde+2)}$.  
Thus, with probability
$\frac{1}{4}\ell^{-(\ctilde+2)}$, $p'$ is at 
least $\frac{1}{4}\ell^{-(\ctilde+2)}$.
Now, after $\BreakOW$ fixed $\ftilde,f,P,S$, in case
$p' \geq \frac{1}{4\ell^{c+2}}$, it is clear that $\BreakOW$
will invert $g$ with this probability (because for any $w$ which is chosen
as $w = g(v)$, $\htilde$ has no preimages of $w$ which $g$ does not have,
and $\BreakOW$ will at least find the preimage $v$ for $\htilde$).
\end{proof}

\begin{lemma}\label{lemma:AisOW}
With probability $1$, the probability that $A(k,\cdot)$ inverts
$f(k,\cdot)$ is a negligible function in $k$.
\end{lemma}
\begin{proof}
Let $B_{k,\alpha}$ be the event that $A(k,\cdot)$ inverts $f(k,\cdot)$ with
probability at least $k^{-\alpha}$.
We show that for any $\alpha \in \bbN$, with probability
$1$, finitely many events $B_{k,\alpha}$ happen.
By the Borel-Cantelli lemma it is enough to show that 
$\sum_{k} \Pr[B_{k,\alpha}] < \infty$ for any $\alpha$.
For this, it is clearly enough to show that $\Pr[B_{k,\alpha}]$ is a negligible
function in $k$ for any $\alpha$.

To show this, we distinguish cases. 
First, consider the case that $f(k,\cdot)$ is picked as a random 
permutation in GenerateOracles, i.e., the case where
where $q_{\ell,n} \leq \ell^{-(\ctilde+2)}$ for all $n$ 
and $k^c = \ell$.

All oracles $\Breaker(\ell,\cdot)$ which $A(k,\cdot)$ can
possibly access are fixed before $f(k,\cdot)$ is chosen,
and so we can ignore them.  
The same holds for all oracles $f(k',\cdot)$ for $k' \leq k$.

However, $A$ can also access $f(k',\cdot)$ for $k < k' \leq k^c$.
These are picked later, and the distribution can
depend on $f(k,\cdot)$.

Luckily, there is only a polynomial number of possibilities 
how the functions $f(k',\cdot)$ for $k < k' \leq k^c$ will be
chosen in the end.
To see that, note that we can specify how all of 
these $f(k',\cdot)$ are chosen by 
specifying \begin{itemize}
\item the integer $\ell \in \{k^c,\ldots,k^{c^2}\}$ for which
the algorithm GenerateOracles uses the else clause, in case
there is one (note that there is at most one)
\item the integer $\ntilde$ which GenerateOracles uses in this case
\item whether GenerateOracles uses BuildPUBreaker or BuildOWBreaker,
\item and which of the at most $\ell^{\ctilde+3}$ iterations is used 
in the end.
\end{itemize}
Once we have specified these numbers, we see that we know which of the
choices for $S_k,P_k,\ftilde$, and so on are used to pick $f(k',\cdot)$
for all these $k'$.

We can now simply check whether $A(k,\cdot)$ inverts $f(k,\cdot)$
with probability $k^{-\alpha}$ for \emph{any} of these random choices.
Since this probability is negligible, we apply the union bound
and get the result in this case.

The same argument works in case $f(k,\cdot)$ is picked from $\Pi_{n(k)}$
in either BuildOWBreaker or BuildPUBreaker (because $n(k) < \ntilde$).

Thus, consider the last case where $f(k,\cdot)$ is set 
to $S_k \circ \ftilde \circ P_k$
in either BuildOWBreaker or BuildPUBreaker.
Then, for any intertion we consider the breaker which tries 
to invert $\ftilde$ by first inverting $S_k$, 
then running $A(k,\cdot)$, and then applying $P_k$ on 
the result.
The probability that this algorithm inverts $\ftilde$ in any of the at 
most $\ell^{c}$ iterations of BuildPUBreaker or BuildOWBreaker is 
negligible (by the previous sections), and so we get the result in this
case as well.
\end{proof}

\paragraph{Finishing the proof}
We can now finish the proof of Theorem~\ref{thm:mainC}.

First, we see that with probability $1$ the oracles $(f,\Breaker)$ 
generated are such that $\Breaker$ either infinitely often
breaks the one-wayness or the pseudouniformity of $g$:
first, due to Lemma~\ref{lemma:BreakerIsNotBot} we see that
we will infinitely often attempt to construct Breaker in one of
the two ways, and by either Lemma~\ref{lemma:PUDistinguishes} or
Lemma~\ref{lemma:OWBInverts} we see that the probability that this only
works finitely many times is $0$ (again using Borel-Cantelli).
By Lemma~\ref{lemma:AisOW} we see that~$f$ will be one-way for $A$,
which proves the result.

\section{Acknowledgements}
We thank Colin Zheng for pointing out a mistake in an earlier version
of this paper.

{
\footnotesize
\bibliographystyle{alpha}
\bibliography{../db}
}

\end{document}